\documentclass[journal]{IEEEtran}
\usepackage{cite}
\usepackage{textcomp}
\usepackage{xcolor}
\def\BibTeX{{\rm B\kern-.05em{\sc i\kern-.025em b}\kern-.08em
    T\kern-.1667em\lower.7ex\hbox{E}\kern-.125emX}}
\usepackage{amsmath,amssymb,amsfonts}
\usepackage{algorithmic}
\usepackage{graphicx}
\usepackage{tikz}
\usetikzlibrary{intersections}
\usepackage[caption=false]{subfig}
\usetikzlibrary{calc}
\usepackage{bm}
\usepackage[hyphens]{url}
\usepackage[hidelinks]{hyperref}
\hypersetup{colorlinks = false, urlcolor=black}
\hypersetup{breaklinks=true}

\usepackage{multirow}
\usepackage{pgfplots}
\usepackage{textcomp}
\usepackage{amsthm}
\newtheorem{thm}{Theorem}
\newtheorem{lem}[thm]{Lemma}

\newtheorem{rem}{Remark}
\usetikzlibrary{calc,patterns,arrows,shapes.arrows,intersections}
\usetikzlibrary{decorations}
\usetikzlibrary{arrows}
\usetikzlibrary{shapes.geometric}
\usetikzlibrary{trees}
\usetikzlibrary{arrows,shapes}
\usetikzlibrary{calc,patterns,arrows,shapes.arrows,intersections}
\usetikzlibrary{arrows}
\usetikzlibrary{positioning,chains,fit,shapes,calc}
\usetikzlibrary{positioning}
\usetikzlibrary{decorations.text}
\usetikzlibrary{decorations.pathmorphing}
\usetikzlibrary{shapes.misc, positioning}
\usepackage{xcolor}
\usepackage{stackengine}
\def\delequal{\mathrel{\ensurestackMath{\stackon[1pt]{=}{\scriptstyle\Delta}}}}
\usepackage[utf8]{inputenc}
\usepackage{mathtools}
\usepackage{algorithm2e}
\usepackage{optidef}

\usepgfplotslibrary{polar}
\usepackage{ellipsis}
\usetikzlibrary{calc}
\usetikzlibrary{decorations.pathreplacing,decorations.markings,shapes.geometric}

\DeclareMathOperator*{\minimize}{minimize}

\usepackage{glossaries}
\pgfplotsset{compat=1.18}

\newacronym{iot}{IoT}{Internet of things}
\newacronym{bs}{BS}{base station}
\newacronym{miso}{MISO}{multiple-input single-output}
\newacronym{simo}{SIMO}{single-input multiple-output}
\newacronym{mimo}{MIMO}{multiple-input multiple-output}
\newacronym{siso}{SISO}{single-input single-output}
\newacronym{bps}{bps}{bits per second}
\newacronym{mrt}{MRT}{maximum-ratio transmission}
\newacronym{los}{LoS}{line-of-sight}
\newacronym{awgn}{AWGN}{additive white Gaussian noise}
\newacronym{cmos}{CMOS}{complementary metal-oxide semiconductor}
\newacronym{ecc}{ECC}{error control coding}
\newacronym{adc}{ADC}{analog-to-digital converter}
\newacronym{dac}{DAC}{digital-to-analog converter}
\newacronym{dsp}{DSP}{digital signal processing}
\newacronym{lpf}{LPF}{low-pass filter}
\newacronym{bpf}{BPF}{band-pass filter}
\newacronym{lo}{LO}{local oscillator}
\newacronym{rf}{RF}{radio frequency}
\newacronym{dc}{DC}{direct current}
\newacronym{if}{IF}{intermediate frequency}
\newacronym{pa}{PA}{power amplifier}
\newacronym{wpt}{WPT}{wireless power transfer}
\newacronym{snr}{SNR}{signal-to-noise ratio}
\newacronym{ber}{BER}{bit error rate}
\newacronym{per}{PER}{packet error rate}
\newacronym{qam}{QAM}{quadrature amplitude modulation}
\newacronym{bpsk}{BPSK}{binary phase-shift keying }
\newacronym{qpsk}{QPSK}{quadrature phase-shift keying}
\newacronym{oqpsk}{OQPSK}{offset quadrature phase-shift keying}
\newacronym{psd}{PSD}{power spectral density}
\newacronym{eirp}{EIRP}{effective isotropic radiated power}
\newacronym{ism}{ISM}{industrial, scientific and medical}
\newacronym{fcc}{FCC}{Federal Communications Commission}
\newacronym{lp}{LP}{low-power}
\newacronym{mp}{MP}{medium-power}
\newacronym{hp}{HP}{high-power}
\newacronym{cdf}{CDF}{cumulative distribution function}
\newacronym{nlos}{NLOS}{non-line-of-sight}
\newacronym{ppp}{PPP}{Poisson point process}
\newacronym{pdf}{PDF}{probability density function}
\newacronym{ris}{RIS}{reconfigurable intelligent surface}
\newacronym{csi}{CSI}{channel state information}
\newacronym{af}{AF}{amplify and forward}
\newacronym{eh}{EH}{energy harvesting}
\newacronym{em}{EM}{electromagnetic}
\newacronym{ap}{AP}{access point}
\newacronym{nim}{NIM}{negative-index material}
\newacronym{srr}{SRR}{split ring resonator}
\newacronym{iid}{i.i.d.}{independent and identically distributed}
\newacronym{ml}{ML}{maximum likelihood}
\newacronym{lse}{LSE}{least square error}
\newacronym{mse}{MSE}{mean squared error}
\newacronym{rmse}{RMSE}{root mean squared error}
\newacronym{svd}{SVD}{singular value decomposition}\newacronym{mmse}{MMSE}{minimum mean square error}
\newacronym{swipt}{SWIPT}{simultaneous wireless information and power transfer}
\newacronym{blue}{BLUE}{best linear unbiased estimator}
\newacronym{nap}{NAP}{normalized achieved power}
\newacronym{mnap}{MNAP}{mean normalized achieved power}
\newacronym{star}{STAR}{simultaneously transmitting and reflecting}

\IEEEoverridecommandlockouts

\begin{document}
\bstctlcite{IEEEexample:BSTcontrol}

\title{Energy Harvesting Maximization for Reconfigurable Intelligent Surfaces Using Amplitude Measurements}

	\author{Morteza~Tavana,~\IEEEmembership{Student Member,~IEEE,}
		Meysam~Masoudi,
		and~Emil~Bj\"ornson,~\IEEEmembership{Fellow,~IEEE,}
	\thanks{M. Tavana, and E. Bj\"ornson are with the School of Electrical Engineering and Computer Science, KTH Royal Institute of Technology, Stockholm, Sweden. E-mail: \{morteza2, ~emilbjo\}@kth.se.}
	\thanks{M. Masoudi is with Ericsson, Global AI Accelerator (GAIA) unit, Sweden. E-mail: meysam.masoudi@ericsson.com}
	\thanks{A preliminary version of this paper \cite{Mtavana_Asilomar_2023} has been accepted in the IEEE Asilomar Conference on Signals, Systems, and Computers, 2023.}
    \thanks{This paper was supported by Digital Futures.}
}

\maketitle

\vspace{-5mm}
\begin{abstract}
Energy harvesting can enable a reconfigurable intelligent surface (RIS) to self-sustain its operations without relying on external power sources. In this paper, we consider the problem of energy harvesting for RISs in the absence of coordination with the ambient RF source. We propose a series of sequential phase-alignment algorithms that maximize the received power based on only power measurements. We prove the convergence of the proposed algorithm to the optimal value for the noiseless scenario. However, for the noisy scenario, we propose a linear least squares estimator. We prove that within the class of linear estimators, the optimal set of measurement phases are equally-spaced phases. To evaluate the performance of the proposed method, we introduce a random phase update algorithm as a benchmark. Our simulation results show that the proposed algorithms outperform the random phase update method in terms of achieved power after convergence while requiring fewer measurements per phase update. Using simulations, we show that in a noiseless scenario with a discrete set of possible phase shifts for the RIS elements, the proposed method is sub-optimal, achieving a higher value than the random algorithm but not exactly the maximum feasible value that we obtained by exhaustive search.
\end{abstract}

\begin{IEEEkeywords}
Energy harvesting, phased array, reconfigurable intelligent surface, zero-energy devices.\vspace{0mm}
\end{IEEEkeywords}

\IEEEpeerreviewmaketitle
\section{Introduction}
\IEEEPARstart{F}{uture} wireless networks should provide seamless connectivity for the rapidly growing number of devices and services\cite{iot2025,international2015imt,8808168}. If wireless networks are implemented in the same manner as before, the energy consumption would keep increasing dramatically with the traffic volume. From both carbon footprint and energy consumption perspectives, it is a necessity to develop energy-saving techniques that can be implemented in the network nodes, including low-power and zero-energy devices \cite{MTavana_IoTJ_2021, MTavana_2022_ICC, Mtavana_VTCW_2022}.

Traditional wireless networks have no control of the radio propagation environment. Providing connectivity for regions with low \gls*{snr} comes at the cost of deploying more sophisticated transmission schemes, more radio resources such as antennas and spectrum and consequently consuming more energy \cite{goldsmith2005wireless, 9140329}.

With the emergence of the \glspl*{ris}, several limiting factors associated with the propagation environment can be eliminated. A \gls*{ris}  can manipulate the propagation environment to increase the signal strength in the desired direction and guide the \gls*{em} waves toward the receiver via engineered reflections \cite{9048622,6206517,9721205,9206044}. The \gls*{ris} is primarily envisioned for providing coverage for regions that are blocked by objects\cite{9140329}. 

For instance, a \gls*{ris} can be deployed in a city with dense buildings and poor \gls*{los} conditions, or it can also be deployed in homes, where the walls obstruct the signal path.

We can also see the concept of \gls*{ris} as an alternative to traditional relays, which have similar use cases. \glspl*{ris} offer some unique advantages that make them a promising technology for the future of wireless communication systems. \glspl*{ris} can be better than traditional relays in terms of energy efficiency, as they do not require power-hungry amplifiers and can use low-power hardware components. Moreover, \glspl*{ris} do not introduce significant latency into the communication system, making them well-suited for applications such as augmented reality and the \gls*{iot}. Last but not least, \glspl*{ris} can be integrated into existing wireless communication systems with minimal changes, making them a cost-effective solution for improving the performance of existing networks.\footnote{Relays of the repeater type have the latter two advantages, but have a substantially higher energy consumption than \glspl*{ris}.} These advantages make \glspl*{ris} a promising technology for improving the performance of wireless communication systems \cite{8319526,8741198}.  However, \glspl*{ris} come with their own set of challenges, such as limited bandwidth, the risk of generating interference to some users, cumbersome configuration, and a higher cost of deployment and maintenance. Ultimately, the choice between traditional relays and \glspl*{ris} depends on a careful weighing of these factors and an understanding of the specific needs and goals of each situation.

In the presence of a wired power supply, \gls*{ris} may not have clear benefits compared to relays and small \glspl*{bs}, and the respective advantages are debatable \cite{Bjornson2020a}. However, if the \gls*{ris} is self-sustainable, it opens up new deployment possibilities where there is no competition. This is where \gls*{eh} comes into play. By harvesting energy from \gls*{rf} signals that are already present in the environment, the \gls*{ris} can operate without relying on external power sources \cite{ntontin2022wireless,9799791,9148892}. This is particularly beneficial in situations where the \gls*{ris} is deployed in remote or inaccessible locations where it can be difficult or expensive to provide a continuous power source. In \cite{8741198}, the authors developed a model for the \gls*{ris} power consumption that is based on the number of elements and their phase resolution capability. Higher phase resolution in the \gls*{ris} increases the complexity and the power consumption \cite{7370753}.

There are other related works.
The study \cite{9214497} considers a hybrid-relaying scheme empowered by a self-sustainable \gls*{ris} to simultaneously improve the downlink energy transmission and uplink information transmission. The authors proposed time-switching and power-splitting schemes for \gls*{ris} operation.
The paper \cite{9895266} proposes a novel framework for \gls*{wpt} system using a \gls*{ris} to improve power transfer efficiency. The proposed framework employs independent beamforming to replace conventional joint beamforming, which results in higher efficiency. 
In \cite{9679387}, a \gls*{ris}-aided simultaneous terahertz information and power transfer is proposed to maximize the sum data rate of the information users while ensuring that the power harvesting requirements of energy users and the \gls*{ris} are met.

In \cite{10000917}, the authors propose a multi-functional \gls*{ris} (MF-\gls*{ris}) that utilizes \gls*{rf} energy to support reflection and amplification of incident signals. The MF-\gls*{ris} mitigates double-fading attenuation experienced by passive \glspl*{ris} \cite{9998527}.
The study \cite{10025760} considers an \gls*{ris}-assisted \gls*{swipt} scheme that integrates the \gls*{ris} with non-coherent differential chaos shift keying to compensate for energy loss. The authors analyze the \gls*{ber} and energy shortage probability of the scheme over the multipath Rayleigh fading channel. 
Finally, in \cite{9149709}, the authors consider a \gls*{ris}-assisted data transmission from a multi-antenna \gls*{ap} to a receiver. The system operation is cyclic, and in each cycle, there are two operational phases: 1) the \gls*{ris} array harvests energy from the transmitter. 2) In return, \gls*{ris} assists the \gls*{ap} to transmit its data to the receiver.

The state-of-the-art techniques consider perfect \gls*{csi} at the \gls*{ris}, which is obtained by coordination between the transmitters and the \gls*{ris}. The existing solutions require extra hardware (i.e., multiple \gls*{rf} receivers to measure both amplitude and phase to obtain \gls*{csi}) and signaling, which in turn increases the energy consumption and cost of the \gls*{ris}. 

We consider a different scenario, where the \gls*{ris} must configure itself without any \gls*{rf} receivers. The proposed method makes power measurements in the \gls*{eh} units and uses them to maximize the harvesting power iteratively. 

Our preliminary results in \cite{Mtavana_Asilomar_2023} present an amplitude-based sequential optimization of energy harvesting with \glspl*{ris} for the noiseless scenario using continuous phase control. In this paper, we analyze the impact of noise on the phase alignment algorithm and present a linear phase estimator that can be embedded in our proposed sequential phase alignment algorithm. We derive the optimal set of measurement phases within the class of linear estimators. Also, we formulate the problem for a discrete phase control scheme in a noiseless scenario. The main contributions of this paper, as compared to the previous works, are as follows:

\begin{itemize}
    \item We propose a series of sequential phase-alignment algorithms to maximize the received power at the harvesting unit based on power measurements, in different situations.
    \item We prove the optimality of the proposed algorithm analytically for the noiseless scenario with continuous phase control.
    \item We formulate the corresponding problem for a noiseless scenario with discrete phase control and propose a heuristic sequential phase alignment algorithm inspired by the continuous phase control algorithm. 
    \item We extend the result to noisy scenarios and propose a simple linear estimator, that can be embedded in our proposed sequential phase alignment algorithm.
    We optimize the phase measurement configurations for the proposed linear estimator, and we prove that equally spaced phases are optimal.
    \item Our simulation results show that the proposed algorithm greatly outperforms the benchmark random phase-update method in terms of the number of required measurements to achieve the optimum.
\end{itemize}

The remainder of this paper is organized as follows. Section~\ref{sec:PF} describes the \gls*{ris} hardware architecture, two different \gls*{eh} schemes in the \gls*{ris}, and the phase alignment problem. Section \ref{Sec:Single_phase} presents the proposed phase estimators for the received power maximization by tuning the phase of a \gls*{ris} element.
Section \ref{sec:PS} describes the proposed abstract model of the \gls*{ris} operation and the proposed solution for continuous and discrete phase control scenarios. Simulation results are presented in Section \ref{sec:NumResults}, while Section \ref{sec:Conclusion} provides our conclusions.

\textit{Notations}: We denote sets by upper-case script letters, e.g., $\mathcal{S}$, or upper-case Greek letters, e.g., $\Omega$. The only exceptions are the sets of natural numbers, real numbers, and complex numbers that we represent with $\mathbb{N}$, $\mathbb{R}$, and $\mathbb{C}$, respectively.  The cardinality of a set $\mathcal{S}$ is represented by $|\mathcal{S}|$. Vectors are indicated by lower-case bold-face letters, e.g., $\bm{x}$, and $x_{i}$ denotes the $i$th element of $\bm{x}$. We represent matrices by upper-case bold-face letters, e.g., $\bm{A}$, and $\left[\bm{A}\right]_{m,n}$ indicates the element of $\bm{A}$ with row number $m$ and column number $n$. The identity matrix of size $n$ is denoted by $\bm{I}_{n}$. The expectation and covariance operators are represented by $\mathbb{E}{\left(\cdot\right)}$ and $\mathbb{C}\mathrm{ov}{\left(\cdot\right)}$, respectively. Also, $\mathbb{E}_{\mathrm{X}}{\left(\cdot\right)}$ denotes the expectation operation with respect to the random variable $\mathrm{X}$. We represent random variables with upright upper-case letters, e.g., $\mathrm{X}$, and we denote random vectors with upright upper-case bold-face letters, e.g., $\bm{\mathrm{X}}$.   Also, we use $p_{\mathrm{X}}{\left(x\right)}$ to indicate the \gls*{pdf} of a continuous random variable $\mathrm{X}$ at $x$. We use $\bm{1}$ to represent the all-one vector. The curled inequality $\succeq $ indicates componentwise inequality between vectors. We also use $\delequal$ to indicate an equal by definition sign. With $\operatorname{tr}\!{\left(\cdot\right)}$ and $\left(\cdot\right)^{\mathsf{T}}$, we denote the trace and transpose operators. The operator $\operatorname{diag}\!{\left(\bm{x}\right)}$ returns a square diagonal matrix with the elements of vector $\bm{x}$ on the main diagonal. We denote the imaginary unit by $j\delequal\sqrt{-1}$. We represent the conjugate of a complex number $z$ with $z^{*}$. However, we denote the optimal solution with the superscript ${\star}$, e.g., $x^{\star}$. Also, the operation $\operatorname{Arg}{\left(z\right)}$ returns the principal value of the argument of $z$ that lies within the interval $(-\pi, \pi]$, while $\arg{\left(z\right)}$ returns the set of all possible values of the argument of $z$.

\section{Problem Description}\label{sec:PF}
This section presents 1) the \gls*{ris} device architecture, 2) the \gls*{eh} schemes, and 3) the phase alignment problem for \gls*{ris}-assisted \gls*{eh}.

\subsection{RIS Hardware Architecture}
Fig.~\ref{fig:RIS_structure} illustrates the hardware architecture of a typical \gls*{ris}. The front layer consists of several metal elements that are printed on a dielectric substrate and arranged in a two-dimensional array to reflect the incoming signals. The size of the elements is typically smaller than the wavelength of the signal of interest \cite{chen2016review, he2019tunable,9306896}. For each patch element, there is a controllable circuit that is used to adjust the phase of the reflected signals and direct the signal in the desired direction. 

There is a copper layer behind the dielectric substrate to reduce the \gls*{rf} waves leakage from the back \cite{8910627, 9161157}. There is a control circuit board that can control the reflection amplitudes and phase shifts of the patch elements. The \gls*{ris} controller can communicate with the network components such as \glspl*{bs} via a connectivity interface to receive the reflection state configuration \cite{9140329, wu2021intelligent, tsilipakos2020toward}.  Finally, the \gls*{ris} requires a power supply to adjust the phases and then maintain the desired reflection state.

\begin{figure}[t]
    \vspace{-10mm}
    \hspace{-21mm}
    \scalebox{0.75}{\definecolor{mycolor1}{rgb}{0.85,0.325,0.098}
\definecolor{mycolor2}{rgb}{0.00000,0.44700,0.74100}
\definecolor{mycolor3}{rgb}{0.00000,0.49804,0.00000}
\definecolor{apricot}{rgb}{0.98, 0.81, 0.69}
\definecolor{copper}{rgb}{0.72, 0.45, 0.2}
\definecolor{babyblue}{rgb}{0.54, 0.81, 0.94}
\definecolor{gray}{rgb}{0.5, 0.5, 0.5}
\usetikzlibrary{decorations.pathreplacing}

\tikzset{radiation/.style={{decorate,decoration={expanding waves,angle=90,segment length=4pt}}},
         relay/.pic={
        code={\tikzset{scale=5/10}
            \draw[semithick] (0,0) -- (1,4);
            \draw[semithick] (3,0) -- (2,4);
            \draw[semithick] (0,0) arc (180:0:1.5 and -0.5) node[above, midway]{#1};
            \node[inner sep=4pt] (circ) at (1.5,5.5) {};
            \draw[semithick] (1.5,5.5) circle(8pt);
            \draw[semithick] (1.5,5.5cm-8pt) -- (1.5,4);
            \draw[semithick] (1.5,4) ellipse (0.5 and 0.166);
            \draw[semithick,radiation,decoration={angle=45}] (1.5cm+8pt,5.5) -- +(0:2);
            \draw[semithick,radiation,decoration={angle=45}] (1.5cm-8pt,5.5) -- +(180:2);
  }}
}

\begin{tikzpicture}[cross/.style={path picture={ 
  \draw[black]
(path picture bounding box.south east) -- (path picture bounding box.north west) (path picture bounding box.south west) -- (path picture bounding box.north east);
}}]
\pgfmathsetmacro{\R}{0.5}
\pgfmathsetmacro{\W}{1.25}
\pgfmathsetmacro{\H}{1.25}
\draw[fill=gray] (-7.5,1.73) rectangle (-6,3.21) node[white] at (-6.75,2.32) {\small Controller};
\node[white] at (-6.75,2.69) {\small RIS};
\foreach \i in {1,...,14} {
    \draw[] (-6,3.23-0.1*\i)-- (-5.9,3.23-0.1*\i);
    \draw[] (-7.5,3.23-0.1*\i)-- (-7.6,3.23-0.1*\i);
    \draw[] (-6-0.1*\i,3.23)-- (-6-0.1*\i,3.33);
    \draw[] (-6-0.1*\i,1.73)-- (-6-0.1*\i,1.63);}
\draw [stealth-stealth,thick] (-5.9,2.48) -- (-3.96,2.48);
\path (-7.125,-1.7)  pic[scale=0.5,color=black] {relay={\small BS}};
\draw [stealth-stealth,thick] (-6.75,1.63) -- (-6.75,0);
\filldraw[fill=babyblue, draw=babyblue] (-3.14-0.82,-1.5+0.82) rectangle (0.88-0.82,2.52+0.82) node[] at (-1.2-0.82,2.15+0.82) {\small Control circuit board};
\filldraw[fill=copper, draw=copper] (-3.14,-1.5) rectangle (0.88,2.52) node[] at (-1.2,2.15) {\small Copper};
\filldraw[fill=apricot, draw=apricot] (-2.32,-2.32) rectangle (1.7,1.7);
\foreach \i in {-3,...,2} {
    \foreach \j in {-3,...,2} {
    \node [draw, copper, line width=2,trapezium, trapezium left angle=90, trapezium right angle=90, minimum width=\R cm,outer sep=0pt,fill=apricot, rotate = 0] at (\i*\W*\R,\j*\H*\R) {};
    \draw[copper, line width=2] (\i*\W*\R,\j*\H*\R+0.5*\R) -- (\i*\W*\R,\j*\H*\R+0.15*\R);
    \draw[copper, line width=2] (\i*\W*\R,\j*\H*\R-0.5*\R) -- (\i*\W*\R,\j*\H*\R-0.15*\R);
    \draw[copper, line width=2] (\i*\W*\R-0.25*\R,\j*\H*\R-0.15*\R) -- (\i*\W*\R+0.25*\R,\j*\H*\R-0.15*\R);
    \draw[copper, line width=2] (\i*\W*\R-0.25*\R,\j*\H*\R+0.15*\R) -- (\i*\W*\R+0.25*\R,\j*\H*\R+0.15*\R);}}
    
\begin{polaraxis}[at={(-3.87in,-1.65in)},
		grid=none,
		width=9cm,
		xtick=\empty,
		ytick=\empty,
		rotate=5,
		axis line style={draw=none},
		tick style={draw=none}]
		\addplot+[mark=none, fill=mycolor3, opacity=0.4, line width=0.75pt, color = mycolor3, domain=-180:180,samples=600] 
		{0.8*abs(sin(2*x*3.141592)/(3.141592*3.141592/180*x))}; 
\end{polaraxis}

\begin{polaraxis}[at={(-1.9in,-1.42in)},
		grid=none,
		width=9cm,
		xtick=\empty,
		ytick=\empty,
		rotate=30,
		axis line style={draw=none},
		tick style={draw=none}]
		\addplot+[mark=none, fill=mycolor3, opacity=0.4, line width=0.75pt, color = mycolor3, domain=-180:180,samples=600] 
		{0.1*abs(sin(4*x*3.141592)/(3.141592*3.141592/180*x))}; 
\end{polaraxis}

\begin{polaraxis}[at={(-1.84in,-1.65in)},
		grid=none,
		width=9cm,
		xtick=\empty,
		ytick=\empty,
		rotate=-15,
		axis line style={draw=none},
		tick style={draw=none}]
		\addplot+[mark=none, fill=mycolor3, opacity=0.4, line width=0.75pt, color = mycolor3, domain=-180:180,samples=600] 
		{0.1*abs(sin(4*x*3.141592)/(3.141592*3.141592/180*x))}; 
\end{polaraxis}

\begin{polaraxis}[at={(-1.89in,-1.94in)},
		grid=none,
		width=9cm,
		xtick=\empty,
		ytick=\empty,
		rotate=-60,
		axis line style={draw=none},
		tick style={draw=none}]
		\addplot+[mark=none, fill=mycolor3, opacity=0.4, line width=0.75pt, color = mycolor3, domain=-180:180,samples=600] 
		{0.1*abs(sin(4*x*3.141592)/(3.141592*3.141592/180*x))}; 
\end{polaraxis}

\end{tikzpicture}}
    \vspace{-15mm}
    \caption{Hardware structure of a RIS.}
    \label{fig:RIS_structure}
\end{figure}
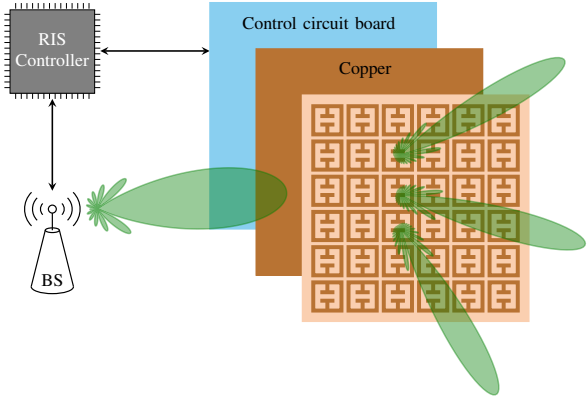
\subsection{Energy Harvesting Schemes}
In general, an \gls*{ris} can harvest \gls*{rf} energy directly or indirectly from an ambient \gls*{rf} source. We will describe these scenarios below and later show that they lead to system models of the same kind. 

\subsubsection{Direct EH}
The \gls*{ris} elements are capable of both reflecting and receiving the incident \gls*{em} waves. During the harvesting phase, the \gls*{ris} operates in the reception mode, where it combines the received signals from each element with some phase shifts.\footnote{There is a type of metasurface implementation called holographic beamforming that can pass the incident EM waves from one side to the other. The metamaterial can add adjustable phase shifts, similar to a phased array, but with a different implementation \cite{black2017holographic}.} The \gls*{ris} can use the harvested energy to fully sustain its operation or, if the energy is insufficient, it can decrease the consumption from other energy sources, as a first step toward achieving a zero-energy \gls*{ris} system. The abstract model of the direct \gls*{eh} operation is shown in Fig.~\ref{fig:direct_RIS}.

\subsubsection{Indirect EH}
The \gls*{ris} is generally deployed and designed to reflect the \gls*{em} waves towards a desired location (typically a receiver location). Inspired by this principle, the \gls*{eh} device can alternatively be deployed in front of the \gls*{ris} (at a short distance), and the phases of the \gls*{ris} elements can be aligned so they combine constructively at the location of the \gls*{eh} device. This device then can return the energy to the \gls*{ris} via a cable. The concept of indirect \gls*{eh} is demonstrated in Fig.~\ref{fig:Indirect_RIS}.

\begin{figure*}[!t]
          \hspace{-14mm}
          \captionsetup[subfigure]{oneside,margin={2cm,0cm}}
          \subfloat[Direct \gls*{eh}.]{\scalebox{0.72}{\definecolor{mycolor1}{rgb}{0.85,0.325,0.098}
\definecolor{mycolor2}{rgb}{0.00000,0.44700,0.74100}
\definecolor{mycolor3}{rgb}{0.00000,0.49804,0.00000}
\definecolor{carnationpink}{rgb}{1.0, 0.65, 0.79}
\definecolor{carnelian}{rgb}{0.7, 0.11, 0.11}
\definecolor{caputmortuum}{rgb}{0.35, 0.15, 0.13}
\definecolor{apricot}{rgb}{0.98, 0.81, 0.69}
\definecolor{copper}{rgb}{0.72, 0.45, 0.2}
\definecolor{babyblue}{rgb}{0.54, 0.81, 0.94}
\definecolor{gray}{rgb}{0.5, 0.5, 0.5}
\usetikzlibrary{decorations.pathreplacing}
\linespread{0.6}
\tikzstyle{startstop} = [rectangle, rounded corners, minimum width=1.5cm, minimum height=1cm,text centered, draw=black, fill=gray!20, inner sep=-0.3ex]

\tikzstyle{arrow} = [thick,->,>=stealth]

\begin{tikzpicture}[cross/.style={path picture={ 
  \draw[black]
(path picture bounding box.south east) -- (path picture bounding box.north west) (path picture bounding box.south west) -- (path picture bounding box.north east);
}}]

\tikzset{radiation/.style={{decorate,decoration={expanding waves,angle=90,segment length=4pt}}},
         relay/.pic={
        code={\tikzset{scale=5/10}
            \draw[thick] (0,0) -- (1,4);
            \draw[thick] (3,0) -- (2,4);
            \draw[thick] (0,0) arc (180:0:1.5 and -0.5) node[above, midway]{#1};
            \node[inner sep=4pt] (circ) at (1.5,5.5) {};
            \draw[thick] (1.5,5.5) circle(8pt);
            \draw[thick] (1.5,5.5cm-8pt) -- (1.5,4);
            \draw[thick] (1.5,4) ellipse (0.5 and 0.166);
            \draw[thick,radiation,decoration={angle=45}] (1.5cm+8pt,5.5) -- +(0:2);
            \draw[thick,radiation,decoration={angle=45}] (1.5cm-8pt,5.5) -- +(180:2);
  }}
}

\pgfmathsetmacro{\R}{0.375}
\pgfmathsetmacro{\W}{1.25}
\pgfmathsetmacro{\H}{1.25}
\pgfmathsetmacro{\RISx}{-5.455}
\pgfmathsetmacro{\RISy}{3.6}

\filldraw[fill=apricot, draw=apricot] (\RISx-2.11,\RISy-1.65) rectangle (\RISx+1.64,\RISy+1.18);
\foreach \i in {-4,...,3} {
    \foreach \j in {-3,...,2} {
    \node [draw, caputmortuum, line width=1,trapezium, trapezium left angle=90, trapezium right angle=90, minimum width=\R cm,outer sep = 0pt,fill=apricot, rotate = 0] at (\RISx+\i*\W*\R,\RISy+\j*\H*\R) {};
    \draw[caputmortuum, line width=1] (\RISx+\i*\W*\R,\RISy+\j*\H*\R+0.5*\R) -- (\RISx+\i*\W*\R,\RISy+\j*\H*\R+0.15*\R);
    \draw[caputmortuum, line width=1] (\RISx+\i*\W*\R,\RISy+\j*\H*\R-0.5*\R) -- (\RISx+\i*\W*\R,\RISy+\j*\H*\R-0.15*\R);
    \draw[caputmortuum, line width=1] (\RISx+\i*\W*\R-0.25*\R,\RISy+\j*\H*\R-0.15*\R) -- (\RISx+\i*\W*\R+0.25*\R,\RISy+\j*\H*\R-0.15*\R);
    \draw[caputmortuum, line width=1] (\RISx+\i*\W*\R-0.25*\R,\RISy+\j*\H*\R+0.15*\R) -- (\RISx+\i*\W*\R+0.25*\R,\RISy+\j*\H*\R+0.15*\R);}}

\path (-11,1)  pic[scale=0.5,color=black] {relay={TX}};

\node[text width=5.5em]  (SC)  at (-1,3.363) [startstop] {Signal\\\vspace{1mm}combiner};

\node[text width=5.5em]  (PS)  at (-5.69,6) [startstop] {Power\\\vspace{1mm}supply};

\node[text width=5.5em]  (EH)  at (-1,6) [startstop] {EH device};

\draw[arrow, line width=0.5mm, mycolor2] (-3.812,3.363) -- (SC.west);

\draw[arrow, line width=0.5mm, mycolor2] (SC.north) -- (EH.south);
\draw[arrow, line width=0.5mm, mycolor2] (EH.west) -- (PS.east);

\draw[arrow, line width=0.5mm, mycolor2] (PS.south) --++ (0,-0.7);

\node[caputmortuum] at (-5.7, 1.5) {\Large Reception mode};

\node[black] at (-9, 1) {$\bm{h}\delequal \begin{bmatrix}
h_{1} \\
\vdots \\
h_{N}
\end{bmatrix}$};


\node[black] at (-2.1, 5.3) {\small Measurements};

\draw[arrow, dotted, line width=0.5mm, mycolor2] (-1.95, 5.75) --++ (-1.22,0) --++ (0,-0.885);


\pgfmathsetmacro{\ControllerX}{2.52}
\pgfmathsetmacro{\ControllerY}{3.15}
\begin{pgflowlevelscope}{\pgftransformscale{0.75}}
\draw[fill=gray] (\ControllerX-7.5,\ControllerY+1.73) rectangle (\ControllerX-6,\ControllerY+3.23) node[white] at (\ControllerX-6.75,\ControllerY+2.30) {\small Controller};
\node[white] at (\ControllerX-6.75,\ControllerY+2.69) {\small RIS};
\foreach \i in {1,...,14} {
    \draw[] (\ControllerX-6,\ControllerY+3.23-0.1*\i)-- (\ControllerX-5.9,\ControllerY+3.23-0.1*\i);
    \draw[] (\ControllerX-7.5,\ControllerY+3.23-0.1*\i)-- (\ControllerX-7.6,\ControllerY+3.23-0.1*\i);
    \draw[] (\ControllerX-6-0.1*\i,\ControllerY+3.23)-- (\ControllerX-6-0.1*\i,\ControllerY+3.33);
    \draw[] (\ControllerX-6-0.1*\i,\ControllerY+1.73)-- (\ControllerX-6-0.1*\i,\ControllerY+1.63);}
\end{pgflowlevelscope}

\begin{polaraxis}[at={(-4.61in,0.2in)},
		grid=none,
		width=5.5cm,
		xtick=\empty,
		ytick=\empty,
		rotate=15,
		axis line style={draw=none},
		tick style={draw=none}]
		\addplot+[mark=none, fill=mycolor3, opacity=0.4, line width=0.75pt, color = mycolor3, domain=-180:180,samples=600] 
		{0.1*abs(sin(0.5*x*3.141592)/(3.141592*3.141592/180*x))}; 
\end{polaraxis}

\end{tikzpicture}}%
              \label{fig:direct_RIS}}
          \hfil
          \hspace{-6mm}
          \subfloat[Indirect  \gls*{eh}.]{\scalebox{0.72}{\definecolor{mycolor1}{rgb}{0.85,0.325,0.098}
\definecolor{mycolor2}{rgb}{0.00000,0.44700,0.74100}
\definecolor{mycolor3}{rgb}{0.00000,0.49804,0.00000}

\definecolor{apricot}{rgb}{0.98, 0.81, 0.69}
\definecolor{copper}{rgb}{0.72, 0.45, 0.2}
\definecolor{babyblue}{rgb}{0.54, 0.81, 0.94}
\definecolor{gray}{rgb}{0.5, 0.5, 0.5}
\usetikzlibrary{decorations.pathreplacing}
\linespread{0.6}
\tikzstyle{startstop} = [rectangle, rounded corners, minimum width=1.5cm, minimum height=1cm,text centered, draw=black, fill=gray!20, inner sep=-0.3ex]

\tikzstyle{arrow} = [thick,->,>=stealth]

\begin{tikzpicture}[cross/.style={path picture={ 
  \draw[black]
(path picture bounding box.south east) -- (path picture bounding box.north west) (path picture bounding box.south west) -- (path picture bounding box.north east);
}}]

\tikzset{radiation/.style={{decorate,decoration={expanding waves,angle=90,segment length=4pt}}},
         relay/.pic={
        code={\tikzset{scale=5/10}
            \draw[thick] (0,0) -- (1,4);
            \draw[thick] (3,0) -- (2,4);
            \draw[thick] (0,0) arc (180:0:1.5 and -0.5) node[above, midway]{#1};
            \node[inner sep=4pt] (circ) at (1.5,5.5) {};
            \draw[thick] (1.5,5.5) circle(8pt);
            \draw[thick] (1.5,5.5cm-8pt) -- (1.5,4);
            \draw[thick] (1.5,4) ellipse (0.5 and 0.166);
            \draw[thick,radiation,decoration={angle=45}] (1.5cm+8pt,5.5) -- +(0:2);
            \draw[thick,radiation,decoration={angle=45}] (1.5cm-8pt,5.5) -- +(180:2);
  }}
}

\pgfmathsetmacro{\R}{0.375}
\pgfmathsetmacro{\W}{1.25}
\pgfmathsetmacro{\H}{1.25}
\pgfmathsetmacro{\RISx}{-5.455}
\pgfmathsetmacro{\RISy}{3.6}

\filldraw[fill=apricot, draw=apricot] (\RISx-2.11,\RISy-1.65) rectangle (\RISx+1.64,\RISy+1.18);
\foreach \i in {-4,...,3} {
    \foreach \j in {-3,...,2} {
    \node [draw, copper, line width=1,trapezium, trapezium left angle=90, trapezium right angle=90, minimum width=\R cm,outer sep = 0pt,fill=apricot, rotate = 0] at (\RISx+\i*\W*\R,\RISy+\j*\H*\R) {};
    \draw[copper, line width=1] (\RISx+\i*\W*\R,\RISy+\j*\H*\R+0.5*\R) -- (\RISx+\i*\W*\R,\RISy+\j*\H*\R+0.15*\R);
    \draw[copper, line width=1] (\RISx+\i*\W*\R,\RISy+\j*\H*\R-0.5*\R) -- (\RISx+\i*\W*\R,\RISy+\j*\H*\R-0.15*\R);
    \draw[copper, line width=1] (\RISx+\i*\W*\R-0.25*\R,\RISy+\j*\H*\R-0.15*\R) -- (\RISx+\i*\W*\R+0.25*\R,\RISy+\j*\H*\R-0.15*\R);
    \draw[copper, line width=1] (\RISx+\i*\W*\R-0.25*\R,\RISy+\j*\H*\R+0.15*\R) -- (\RISx+\i*\W*\R+0.25*\R,\RISy+\j*\H*\R+0.15*\R);}}

\path (-11,1)  pic[scale=0.5,color=black] {relay={TX}};

\node[text width=5.5em]  (EH)  at (-0.5,2.5) [startstop] {EH device};

\node at (-2.6, 6.35) {Power transfer via cable};

\node[copper] at (-5.7, 1.5) {\Large Reflection mode};

\node[text width=5.5em]  (PS)  at (-5.69,6) [startstop] {Power\\\vspace{1mm}supply};

\draw[arrow, line width=0.5mm, mycolor2] (EH.north) |- (PS.east);

\draw[arrow, line width=0.5mm, mycolor2] (PS.south) --++ (0,-0.7);

\node[black] at (-9, 1) {$\bm{h}\delequal \begin{bmatrix}
h_{1} \\
\vdots \\
h_{N}
\end{bmatrix}$};

\node[black] at (-2.7, 1.5) {$\bm{g}\delequal \begin{bmatrix}
g_{1} \\
\vdots \\
g_{N}
\end{bmatrix}$};


\node[black] at (-1.53, 4.6) {\small Measurements};


\draw[arrow, dotted, line width=0.5mm, mycolor2] (-0.9, 3) --++ (0, 1.22) --++ (-1.63,0);


\pgfmathsetmacro{\ControllerX}{2.52}
\pgfmathsetmacro{\ControllerY}{3.15}
\begin{pgflowlevelscope}{\pgftransformscale{0.75}}
\draw[fill=gray] (\ControllerX-7.5,\ControllerY+1.73) rectangle (\ControllerX-6,\ControllerY+3.23) node[white] at (\ControllerX-6.75,\ControllerY+2.30) {\small Controller};
\node[white] at (\ControllerX-6.75,\ControllerY+2.69) {\small RIS};
\foreach \i in {1,...,14} {
    \draw[] (\ControllerX-6,\ControllerY+3.23-0.1*\i)-- (\ControllerX-5.9,\ControllerY+3.23-0.1*\i);
    \draw[] (\ControllerX-7.5,\ControllerY+3.23-0.1*\i)-- (\ControllerX-7.6,\ControllerY+3.23-0.1*\i);
    \draw[] (\ControllerX-6-0.1*\i,\ControllerY+3.23)-- (\ControllerX-6-0.1*\i,\ControllerY+3.33);
    \draw[] (\ControllerX-6-0.1*\i,\ControllerY+1.73)-- (\ControllerX-6-0.1*\i,\ControllerY+1.63);}
\end{pgflowlevelscope}

\begin{polaraxis}[at={(-2.55in,0.2in)},
		grid=none,
		width=6.6cm,
		xtick=\empty,
		ytick=\empty,
		rotate=-21,
		axis line style={draw=none},
		tick style={draw=none}]
		\addplot+[mark=none, fill=mycolor3, opacity=0.4, line width=0.75pt, color = mycolor3, domain=-180:180,samples=600] 
		{0.08*abs(sin(4*x*3.141592)/(3.141592*3.141592/180*x))}; 
\end{polaraxis}

\begin{polaraxis}[at={(-4.61in,0.2in)},
		grid=none,
		width=5.5cm,
		xtick=\empty,
		ytick=\empty,
		rotate=15,
		axis line style={draw=none},
		tick style={draw=none}]
		\addplot+[mark=none, fill=mycolor3, opacity=0.4, line width=0.75pt, color = mycolor3, domain=-180:180,samples=600] 
		{0.1*abs(sin(0.5*x*3.141592)/(3.141592*3.141592/180*x))}; 
\end{polaraxis}

\end{tikzpicture}}%
             \label{fig:Indirect_RIS}}
          \caption{Energy harvesting schemes.}
          \label{fig:direct_Indirect_RIS}
\end{figure*}
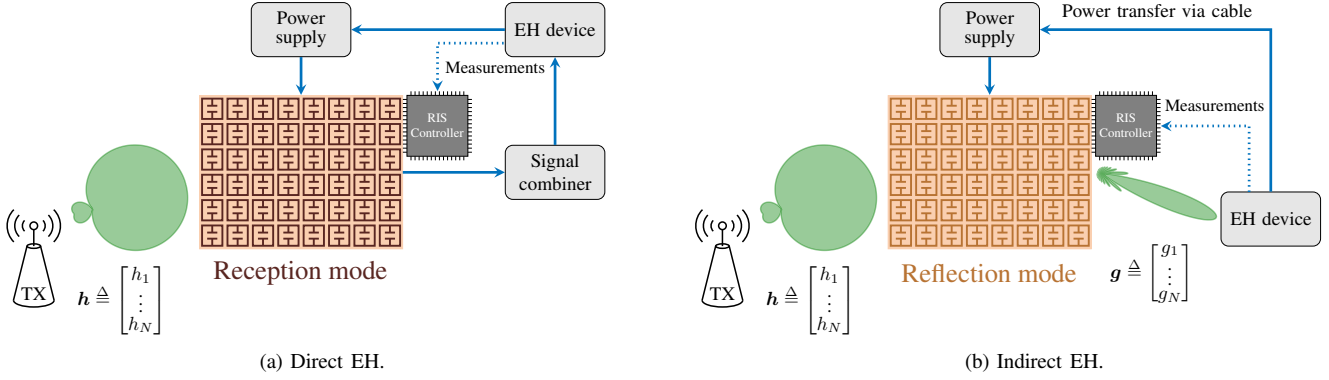

The \gls*{eh} device is equipped with an \gls*{rf} power sensor to measure the power level of a signal. There are different types of \gls*{rf} sensors. It might be based on semiconductor devices like diodes that respond to changes in input power. The signal is then passed through an \gls*{adc}, which converts the analog \gls*{rf} signal into a digital representation that can be processed by digital circuits and microcontrollers.

The indirect approach is also applicable for \gls*{star} \gls*{ris} \cite{9570143} as long as the harvesting unit is on one side of the \gls*{ris}.

\subsection{Problem Description for Energy Harvesting}
We consider a scenario of \gls*{eh} from an ambient \gls*{rf} source\footnote{The proposed model is valid for multiple transmitters under one of these conditions: 1) During the power measurements for the phase update of an element, the amplitude of the transmitted signals from multiple \gls*{rf}  sources remain constant. 2) All \gls*{rf} sources transmit the same narrowband signal.} by an \gls*{ris} with no prior  \gls*{csi}, and there is no coordination between the \gls*{ris} and the \gls*{rf} transmitter. We assume that the transmit power and the location of the transmitter are unknown.\footnote{This scenario is more practical (compared to a scenario with coordination between the transmitter and the \gls*{ris}) as the transmitter may not be designed to coordinate the \gls*{csi} with the \gls*{ris} or only do it when it requests that the \gls*{ris} is supporting its data transmissions.}

For both the considered schemes, the measured \gls*{rf} power by the \gls*{ris} (\gls*{eh} device) has the following expression
\begin{equation}\label{eq:Power_func}
    \mathrm{Y}=\left|\sum_{n=1}^{N} z_{n}e^{j\vartheta_{n}}+\mathrm{W}\right|^{2},
\end{equation}
where $z_{n}\in\mathbb{C}$ for each $0\leq n\leq N$. In \eqref{eq:Power_func}, $z_{n}$ not only includes all channel gains between transmitters and the energy harvester (except the adjustable phase shift $\vartheta_{n}$), but also takes into account the transmission power. Also, $\vartheta_{n}\in[0, 2\pi)$ is the adjustable phase shift of the $n$th element of the \gls*{ris}. The random variables $\mathrm{W}$ and $\mathrm{Y}$ represent the received noise and the measured received power, respectively \cite{9721205,haykin2008communication}. The measured received power can be obtained from the reading of the power in the \gls*{eh} device.\footnote{Power measurements can be obtained from the input of the \gls*{eh} unit using a circuit such as a voltmeter. Alternatively, power measurements can be applied at the output of the \gls*{eh} unit by compensating for the nonlinear harvesting conversion efficiency.} 

The key difference between the direct and indirect \gls*{eh} schemes is in how the parameters in \eqref{eq:Power_func} are selected. For all $1\leq n\leq N$, we have
\begin{equation} \label{eq:z_definition}
    z_{n}\delequal\begin{cases}
    \sqrt{P_{\text{t}}}h_{n}, \qquad &\text{Direct EH},\\
    \sqrt{P_{\text{t}}}h_{n}g_{n}, \qquad &\text{Indirect EH},
    \end{cases}
\end{equation}
where $h_{n}\in \mathbb{C}$ and $g_{n}\in \mathbb{C}$ are the channel gains from the transmitter to the \gls*{ris} element $n$ and from the \gls*{ris} element $n$ to the \gls*{eh} device in the indirect case, respectively. The transmit power is denoted by $P_{\text{t}}$.

We assume no \gls*{csi} is available at the \gls*{ris} controller (i.e., $z_{n}$ is unavailable to the \gls*{ris}). This scenario is more general compared to a scenario with coordination between the transmitters and the \gls*{ris}, since the transmitter might lack a protocol to reach the \gls*{ris} and can even be unaware of its existence.\footnote{
In general, the transmitter doesn't need to be the \gls*{bs}. The BS can request a particular phase configuration for \gls*{ris}-aided data transmission, but this is not what is needed to enable energy harvesting.
}

The values of $\left\{z_{n}\right\}_{n=1}^{N}$ depend on the geometry and propagation environment and can be estimated by measuring amplitude and phase using \gls*{rf} receiver circuits. Since the \gls*{ris} lack such \gls*{rf} chains, the \gls*{ris} cannot estimate the amplitudes and phases. Hence, the values of $\left\{z_{n}\right\}_{n=1}^N$ will remain unknown to the \gls*{ris}. On the other hand, with power measurement at the \gls*{eh} device, the \gls*{ris} can measure the combined power from all elements for any feasible phase configuration. We will utilize such power measurements in this paper.
The general optimization problem is to maximize the received \gls*{rf} power\footnote{Note that the harvested \gls*{rf} power is a nonlinear function of the received \gls*{rf} power that is called the conversion efficiency function. We consider the problem of choosing optimal phase shift $\bm{\vartheta}^{\star}$ for the \gls*{ris} elements to maximize the harvesting \gls*{rf} power. However, since the conversion efficiency function is generally an increasing function, the optimal solution for maximization of the harvesting \gls*{rf} power is equivalent to the one for the maximization of the measured received \gls*{rf} power.} by finding proper phase shifts  $\bm{\vartheta}{\delequal}\left[\vartheta_{1},\dots,\vartheta_{N}\right]^{\mathsf{T}}$ for the \gls*{ris}.

\section{Proposed Phase Estimation Algorithms} \label{Sec:Single_phase}

In this section, we will explore how we can maximize the received power in \eqref{eq:Power_func} with respect to the \gls*{ris} phase configuration $\bm{\vartheta}$. For given values of $z_{1}$, $\dots$, $z_{N}$, the closed-form optimal phase configuration is known. However, the $z$-coefficients are unavailable in our setup. Hence, we will explore how we can estimate them using a sequence of power measurements. Before considering the general scheme in the next section, we start with the basic single-phase scenario, where we maximize the received power with respect to a single phase coefficient. We assume that $ze^{j\vartheta}$ is one of the terms in \eqref{eq:Power_func}, and $z_{0}$ is the summation of all other terms. Therefore, the received power becomes
\begin{align}\label{eq:Y_def}
    \mathrm{Y}&\delequal \left|z_{0}+z e^{j\vartheta}+\mathrm{W}\right|^{2},
\end{align}
where $\mathrm{W}\sim \mathcal{CN}{\left(0,\sigma^{2}\right)}$ and $\mathrm{Y}$ is a random variable due to the noise. We know that $\mathbb{E}{\left(\mathrm{Y}^{2}\right)}= \left|z_{0}+z e^{j\vartheta}\right|^{2}+\sigma^{2}$. Therefore, the optimal phase shift that maximizes the mean received power is
\begin{equation}\label{eq:opt_theta_3}
    \vartheta^{\star}=\arg{\left(z_{0}\right)}-\arg{\left(z\right)}.
\end{equation}

If we measure the received power for $L$ different phase configurations from the set of measurement phases $\Phi\delequal\left\{\varphi_{1}, \dots, \varphi_{L}\right\}$, the measured received powers become
\begin{align}\label{eq:Y_def_l}
    \mathrm{Y}_{l}&\delequal \left|z_{0}+z e^{j\varphi_{l}}+\mathrm{W}_{l}\right|^{2} \nonumber\\ 
    &= \left(\sqrt{\bm{x}^{\mathsf{T}}\bm{a}_{l}}+\widetilde{\mathrm{W}}_{\text{r},l}\right)^{2}+\widetilde{\mathrm{W}}_{\text{i},l}^{2}, \quad \forall 1\leq l\leq L,
\end{align}
where
\begin{equation}\label{eq:def_x}
    \bm{x}\delequal \left[\left|z_{0}\right|^{2}+\left|z\right|^{2}, 2\operatorname{Re}\!{\left(z_{0}z^{*}\right)}, 2\operatorname{Im}\!{\left(z_{0}z^{*}\right)}\right]^{\mathsf{T}},
\end{equation}
$\bm{a}_{l}\delequal\left[1, \cos{\left(\varphi_{l}\right)},  \sin{\left(\varphi_{l}\right)}\right]^{\mathsf{T}}$, $\mathrm{W}_{l}\sim\mathcal{CN}{\left(0,\sigma^{2}\right)}$, and $\widetilde{\mathrm{W}}_{\text{r},l}$ and $\widetilde{\mathrm{W}}_{\text{i},l}$ are \gls*{iid} random variables following $\mathcal{N}{\left(0, \sigma^{2}/2\right)}$. 

We observe from \eqref{eq:Y_def_l} that the conditional distribution of $\mathrm{Y}_{l}$ given the channel is non-central Chi-squared distributed with two degrees of freedom. Also, $\mathrm{Y}_{1}, \mathrm{Y}_{2}, \dots, \mathrm{Y}_{L}$ are independent random variables. Therefore, $\bm{\mathrm{Y}}\delequal\left[\mathrm{Y}_{1}, \mathrm{Y}_{2}, \dots, \mathrm{Y}_{L}\right]^{\mathsf{T}}$ has the mean 
\begin{equation}\label{eq:Mean_Y}
    \bm{\mu}\delequal\mathbb{E}{\left(\bm{\mathrm{Y}}\right)}= \bm{A}\bm{x}+\sigma^{2}\bm{1},
\end{equation}
and covariance
\begin{equation}\label{eq:Cov_Y}
    \bm{\Sigma}\delequal\mathbb{C}\mathrm{ov}{\left(\bm{\mathrm{Y}}\right)}= 2\sigma^{2}\operatorname{diag}\!{\left(\bm{A}\bm{x}\right)+\sigma^{4}\bm{I}_{L}},
\end{equation}
where 
\begin{equation}\label{eq:Matrix_A}
    \bm{A}\delequal [\bm{a}_{1}, \dots, \bm{a}_{L}]^{\mathsf{T}}.
\end{equation}
Consequently, $\mathrm{Y}_{l}$ has the conditional \gls*{pdf} 
\begin{align}\label{eq:pdf_y}
    p_{\mathrm{Y}_{l}|\bm{\mathrm{X}}}{\left(y|\bm{x}\right)}=&\frac{1}{\sigma^{2}}\exp{\left(-\frac{y+\bm{x}^{\mathsf{T}}\bm{a}_{l}}{\sigma^{2}}\right)}I_{0}{\left(\frac{2\sqrt{\bm{x}^{\mathsf{T}}\bm{a}_{l}y}}{\sigma^{2}}\right)},
\end{align}
for all $1\leq l\leq L$, where $I_{0}{\left(\cdot\right)}$ is the zero-order modified Bessel function of the first kind.

The noncentral Chi-squared \gls*{pdf} is a log-concave function \cite{finner1997log}. Moreover, a composition with affine mapping preserves the concavity \cite{boyd_2004_co}. Hence, $p_{\mathrm{Y}_{l}|\bm{\mathrm{X}}}{\left(y|\bm{x}\right)}$ in \eqref{eq:pdf_y} is log-concave with respect to $\bm{x}$.
If we conduct multiple measurements for the measurement phases from $\Phi$, then we will have the joint \gls*{pdf}
\begin{align}\label{eq:joint_y}
    p_{\bm{\mathrm{Y}}|\bm{\mathrm{X}}}{\left(\bm{y}|\bm{x}\right)}=&\prod_{l=1}^{L}p_{\mathrm{Y}_{l}|\bm{\mathrm{X}}}{\left(y_{l}|\bm{x}\right)}.
\end{align}
The log-concavity is preserved under multiplication. Hence, $p_{\bm{\mathrm{Y}}|\bm{\mathrm{X}}}{\left(\bm{y}|\bm{x}\right)}$ in \eqref{eq:joint_y} is log-concave with respect to $\bm{x}$.

We can estimate $\bm{x}$ using a \gls*{ml} estimator \cite{poor1998introduction} as
\begin{align}\label{eq:x_hat}
    \widehat{\bm{x}} &\delequal \underset{\bm{x}\in \mathcal{D}_{x}}{\mathrm{argmax}}\ p_{\bm{\mathrm{Y}}|\bm{\mathrm{X}}}{\left(\bm{y}|\bm{x}\right)}\nonumber\\
    &=\underset{\bm{x}\in \mathcal{D}_{x}}{\mathrm{argmin}}\!\sum_{l=1}^{L}\left(\!\frac{\bm{x}^{\mathsf{T}}\bm{a}_{l}}{\sigma^{2}}-\log{\!\left(I_{0}{\left(\!\frac{2\sqrt{\bm{x}^{\mathsf{T}}\bm{a}_{l}y_{l}}}{\sigma^{2}}\right)}\!\right)}\!\right),
\end{align}
where  $\mathcal{D}_{x}$ is the domain of feasible values for $\bm{x}$. Notice that, in \eqref{eq:x_hat}, we took the logarithm and reversed the sign. Since $p_{\bm{\mathrm{Y}}|\bm{\mathrm{X}}}{\left(\bm{y}|\bm{x}\right)}$ is log-concave with respect to $\bm{x}$, the objective function in \eqref{eq:x_hat} becomes a convex function with respect to $\bm{x}$. 

If we denote the optimal solution to problem \eqref{eq:x_hat} as $\widehat{\bm{x}}$, an estimate of the optimal phase shift can be calculated as
\begin{equation} \label{eq:est_vartheta}
\widehat{\vartheta}=\arg{\left(\widehat{x}_{2}+j\widehat{x}_{3}\right)}.
\end{equation}
It should be noted that $\widehat{\vartheta}$ is not the \gls*{ml} estimate of the optimal phase shift.\footnote{The optimization problem for the \gls*{ml} estimation of the optimal phase shift is non-convex, hence it was not considered in this analysis.}

\subsection{ML Estimator}
To find the \gls*{ml} estimate of $\bm{x}$, we need to solve the following convex optimization problem:
\begin{subequations} \label{eq:Exact_problem}
    \begin{alignat}{3}
    &\minimize_{\bm{x}} &\quad&  \sum_{l=1}^{L}\left(\!\frac{\bm{x}^{\mathsf{T}}\bm{a}_{l}}{\sigma^{2}}-\log{\!\left(I_{0}{\left(\!\frac{2\sqrt{\bm{x}^{\mathsf{T}}\bm{a}_{l}y_{l}}}{\sigma^{2}}\right)}\!\right)}\!\right) &\label{eq:optProb1_a}\\
    &\operatorname{subject\ to} &      & x_{1}\geq \sqrt{x_{2}^{2}+x_{3}^{2}}, \quad &  \label{eq:constraint1_b}
    \end{alignat}
\end{subequations}
where \eqref{eq:constraint1_b} is due to the definition of $\bm{x}$, and it is a second-order cone constraint; that is, a convex constraint \cite{boyd_2004_co}. Moreover, \eqref{eq:constraint1_b} implies $\bm{x}^{\mathsf{T}}\bm{a}_{l}\geq0$ for all $1\leq l\leq L$. 

One can solve the problem \eqref{eq:Exact_problem} using standard convex optimization algorithms, such as the interior-point method, and then use \eqref{eq:est_vartheta} to estimate the phase shift.\footnote{At high \gls*{snr}, one can use the approximation $I_{0}{\left(x\right)}\approx e^{x}/\sqrt{2\pi x}$ and substitute it in \eqref{eq:optProb1_a}. The approximation makes the optimization problem simpler and numerically tractable, as the logarithmic and exponential terms cancel each other out.}

\subsection{Linear Estimator}
In this section, we propose a linear least squares estimator of $\bm{x}$. From \eqref{eq:Y_def_l}, the squared difference between the received power $\bm{y}$ and the noiseless signal  can be expressed as
$\sum_{l=1}^{L}(y_{l}-\bm{x}^{\mathsf{T}}\bm{a}_{l})^{2}$.
Hence, the least squares estimate is obtained as 
\begin{align}\label{eq:Linear_est}
    \widehat{\bm{x}}&\delequal\underset{\bm{x}}{\mathrm{argmin}}\sum_{l=1}^{L}\left(y_{l}-\bm{x}^{\mathsf{T}}\bm{a}_{l}\right)^{2}=\bm{A}^{\dagger}\bm{y},
\end{align}
where $\bm{A}^{\dagger}$ is the Moore–Penrose inverse of $\bm{A}$ defined as $\bm{A}^{\dagger} \delequal \left(\bm{A}^{\mathsf{T}}\bm{A}\right)^{-1}\bm{A}^{\mathsf{T}}$.

The performance of the linear estimator in \eqref{eq:Linear_est} depends on the choice of the matrix $\bm{A}$, which in turn is a function of the set of measurement phases $\Phi$. Our goal is to determine the best $\Phi$ that minimizes the \gls*{mse} $\mathsf{MSE}\delequal \mathbb{E}{(\|\widehat{\bm{\mathrm{X}}}-\bm{x}\|^{2})}$ of the estimator.

To compute the \gls*{mse} of the estimator for an arbitrary matrix $\bm{A}$, we need to find the expected values of $\mathbb{E}{\left(\widehat{\bm{\mathrm{X}}}\right)}$ and $\mathbb{E}{\left(\widehat{\bm{\mathrm{X}}}^{\mathsf{T}}\widehat{\bm{\mathrm{X}}}\right)}$. The mean of $\widehat{\bm{\mathrm{X}}}$ is
\begin{align} \label{eq:Mean_estimator}
    \mathbb{E}{\left(\widehat{\bm{\mathrm{X}}}\right)}
    =\bm{A}^{\dagger}\mathbb{E}{\left(\bm{\mathrm{Y}}\right)}
=\bm{A}^{\dagger}\left(\bm{A}\bm{x}+\sigma^{2}\bm{1}\right)
    =\bm{x}+\sigma^{2}\bm{A}^{\dagger}\bm{1},
\end{align}
where it can be easily demonstrated that $\bm{A}^{\dagger}\bm{1} = [1, 0, 0]^{\mathsf{T}}$. Thus, the linear estimator proposed in \eqref{eq:Linear_est} is biased for $x_1$ and unbiased for $x_2$ and $x_3$. Despite the bias in $\widehat{x}_1$, the phase estimation is not affected as the optimal phase shift is a function of $x_2$ and $x_3$.

Furthermore, we have
\begin{align} \label{eq:var_estimator}
    \mathbb{E}{\left(\widehat{\bm{\mathrm{X}}}^{\mathsf{T}}\widehat{\bm{\mathrm{X}}}\right)}&=\mathbb{E}{\left(\left(\bm{A}^{\dagger}\bm{\mathrm{Y}}\right)^{\mathsf{T}}\left(\bm{A}^{\dagger}\bm{\mathrm{Y}}\right)\right)}=\mathbb{E}{\left(\bm{\mathrm{Y}}^{\mathsf{T}}\left(\bm{A}^{\dagger}\right)^{\mathsf{T}}\bm{A}^{\dagger}\bm{\mathrm{Y}}\right)}\nonumber\\
    &=\operatorname{tr}\!{\left(\left(\bm{A}^{\dagger}\right)^{\mathsf{T}}\bm{A}^{\dagger}\bm{\Sigma}\right)}+\bm{\mu}^{\mathsf{T}}\left(\bm{A}^{\dagger}\right)^{\mathsf{T}}\bm{A}^{\dagger}\bm{\mu}\nonumber\\
    &=\operatorname{tr}\!{\left(\left(\bm{A}^{\dagger}\right)^{\mathsf{T}}\!\!\bm{A}^{\dagger}\bm{\Sigma}\right)}+\bm{x}^{\mathsf{T}}\bm{x}+2\sigma^{2}x_{1}+\sigma^4\nonumber\\
    &=2\sigma^{2}\operatorname{tr}\!{\!\left(\!\left(\bm{A}^{\dagger}\right)^{\!\mathsf{T}}\!\!\!\bm{A}^{\dagger}\operatorname{diag}\!{\left(\bm{A}\bm{x}\right)}\!\right)}\!+\!\sigma^{4}\operatorname{tr}\!{\!\left(\!\left(\bm{A}^{\dagger}\right)^{\!\mathsf{T}}\!\!\!\bm{A}^{\dagger}\!\right)}\nonumber\\
    &\hspace{4mm} +\bm{x}^{\mathsf{T}}\bm{x}+2\sigma^{2}x_{1}+\sigma^4,
\end{align}
where $\bm{\mu}$ and $\bm{\Sigma}$ are defined in \eqref{eq:Mean_Y} and \eqref{eq:Cov_Y}, respectively. 

Using \eqref{eq:Mean_estimator} and \eqref{eq:var_estimator}, the \gls*{mse} for a given $\bm{x}$ becomes
\begin{align}\label{eq:MSE}
    \mathsf{MSE} &\delequal \mathbb{E}{\left(\left\|\widehat{\bm{\mathrm{X}}}-\bm{x}\right\|^{2}\right)}=\mathbb{E}{\left(\left(\widehat{\bm{\mathrm{X}}}-\bm{x}\right)^{\mathsf{T}}\left(\widehat{\bm{\mathrm{X}}}-\bm{x}\right)\right)}\nonumber\\
    &=\mathbb{E}{\left(\widehat{\bm{\mathrm{X}}}^{\mathsf{T}}\widehat{\bm{\mathrm{X}}}\right)}-2\bm{x}^{\mathsf{T}}\mathbb{E}{\left(\widehat{\bm{\mathrm{X}}}\right)}+\bm{x}^{\mathsf{T}}\bm{x}\nonumber\\
    &=2\sigma^{2}\!\operatorname{tr}\!\!{\left(\!\left(\bm{A}^{\!\dagger}\right)^{\!\!\mathsf{T}}\!\!\!\bm{A}^{\!\dagger}\!\operatorname{diag}\!{\left(\bm{A}\bm{x}\right)}\!\right)}\!\!+\!\sigma^{4}\operatorname{tr}\!\!{\left(\!\left(\bm{A}^{\!\dagger}\right)^{\!\!\mathsf{T}}\!\!\!\bm{A}^{\!\dagger}\!\right)}\!+\!\sigma^{4}\!.
\end{align}
 Since the \gls*{ris} doesn't know the values of $z_{0}$ and $z$, from the \gls*{ris} perspective, $\bm{x}$ is a random variable and therefore should be denoted as $\bm{\mathrm{X}}$. The objective is to find a matrix $\bm{A}$ that minimizes the expected \gls*{mse} with respect to $\bm{\mathrm{X}}$.  Assuming that the random variables $\mathrm{X}_{2}$ and $\mathrm{X}_{3}$ are independent and follow a zero-mean symmetric distribution,\footnote{As $z_0$ and $z$ are unknown to the \gls*{ris}, from the \gls*{ris} perspective, $z_0$ and $z$ are random variables and the relative angle between them follows a symmetric distribution with a mean of zero.} we can express the following for any matrix $\bm{A}$ with the structure given in \eqref{eq:Matrix_A}
\begin{equation}
    \mathbb{E}_{\bm{\mathrm{X}}}{\left(\operatorname{diag}\!{\left(\bm{A}\bm{\mathrm{X}}\right)}\right)}= \mathbb{E}{\left(\mathrm{X}_{1}\right)}\bm{I}_{L}.
\end{equation} 
Therefore,
\begin{align}
    \mathbb{E}_{\bm{\mathrm{X}}}{\left(\mathsf{MSE}\right)} &=\sigma^{2}\!\left(\!\operatorname{tr}\!{\left(\left(\bm{A}^{\dagger}\right)^{\!\mathsf{T}}\!\!\!\bm{A}^{\dagger}\right)}\left(2\mathbb{E}{\left(\mathrm{X}_{1}\right)}\!+\!\sigma^{2}\right)\!+\!\sigma^{2}\!\right),
\end{align}
where based on the definition of $\bm{\mathrm{X}}$ in \eqref{eq:def_x}, we have $\mathbb{E}{\left(\mathrm{X}_{1}\right)}\geq 0$. Hence, to minimize $\mathbb{E}_{\bm{\mathrm{X}}}{\left(\mathsf{MSE}\right)}$ with respect to $\bm{A}$, we should minimize $\operatorname{tr}\!{\left(\!\left(\bm{A}^{\dagger}\right)^{\mathsf{T}}\!\!\bm{A}^{\dagger}\!\right)}$ with respect to $\bm{A}$.

We can decompose $\bm{A}$ using the \gls*{svd} as $\bm{A}=\bm{U}\bm{D}\bm{V}^{\mathsf{T}}$, where $\bm{U}$ is an $L\times 3$ real semi-unitary matrix, $\bm{D}\delequal \operatorname{diag}\!{\left(\left[d_{1}, d_{2}, d_{3}\right]^{\mathsf{T}}\right)}$ is a $3\times 3$ diagonal matrix with non-negative real numbers on the diagonal, $\bm{V}$ is an $3\times 3$ real unitary matrix. Thus, we have
\begin{align} \label{eq:tr_expression}
    \operatorname{tr}\!{\left(\left(\bm{A}^{\dagger}\right)^{\mathsf{T}}\!\!\bm{A}^{\dagger}\right)}&\!=\!\operatorname{tr}\!{\left(\left(\bm{A}^{\mathsf{T}}\bm{A}\right)^{-1}\right)}\!=\!\operatorname{tr}\!{\left(\bm{D}^{-2}\right)}\!=\!\sum_{i=1}^{3}\frac{1}{d_{i}^{2}}.
\end{align}
To minimize the expression in \eqref{eq:tr_expression}, we state the following theorem.
\begin{thm}\label{Theorem_1}
    Assuming $L\geq 3$, and $\bm{\varphi}=\left(\varphi_{1}, \dots, \varphi_{L}\right)$,  and $\bm{A}$ be an $L\times 3$ matrix such that $\left[\bm{A}\right]_{i,1}=1$, $\left[\bm{A}\right]_{i,2}=\cos{\left(\varphi_{i}\right)}$, and $\left[\bm{A}\right]_{i,3}=\sin{\left(\varphi_{i}\right)}$ for all $1\leq i\leq L$. If $\bm{A}$ has nonzero singular values of $d_{1}$, $d_{2}$, $d_{3}$, then
    \begin{align}\label{eq:phi_hat_a}
    \bm{\varphi}^{\star} &\!\delequal\! \underset{\bm{\varphi}}{\mathrm{argmin}}\!\sum_{i=1}^{3}\!\frac{1}{d_{i}^{2}}\!=\!\left[\varphi_{0}, \varphi_{0}\!+\!\tfrac{2\pi}{L}, \ldots, \varphi_{0}\!+\!\tfrac{2\pi\left(L\!-\!1\right)}{L}\right]^{\!\mathsf{T}}\!\!,
    \end{align}
    where $\varphi_{0}\in\mathbb{R}$. Additionally, $d_{1}^{\star}=\sqrt{L}$, $d_{2}^{\star}=d_{3}^{\star}=\sqrt{L/2}$. 
\end{thm}
\begin{proof}\renewcommand{\qedsymbol}{}
The proof is provided in Appendix \ref{App_T1}.
\end{proof}

The optimal $\bm{A}^{\star}$ can be constructed by substituting the optimal measurement phases obtained from Theorem \ref{Theorem_1} into \eqref{eq:Matrix_A}. Hence, we have
\begin{equation} \label{eq:opt_A_dagger}
    \bm{A}^{\star^\dagger}=\frac{1}{L}\operatorname{diag}\!{\left(\left[1, 2, 2\right]^{\mathsf{T}}\right)}\bm{A}^{\star^{\mathsf{T}}},
\end{equation} 
and from \eqref{eq:est_vartheta}, \eqref{eq:Linear_est}, and \eqref{eq:opt_A_dagger}, we have
\begin{align}\label{eq:noiseless_theta_hat}
    \widehat{\vartheta}=\arg{\left(\sum_{l=1}^{L}y_{l}e^{j{2\pi \left(l-1\right)}/{L}}\right)}.
\end{align}

\begin{rem}
  As the covariance of the observation vector $\bm{\mathrm{Y}}$ is a function of the parameter $\bm{x}$ \eqref{eq:Cov_Y}, the conventional \gls*{blue} does not exist \cite{kay1993fundamentals}. However, one can define a modified \gls*{blue}-based estimator that, instead of minimizing the \gls*{mse}, minimizes the expected \gls*{mse} with respect to the parameter $\bm{\mathrm{X}}$. Assuming that $\mathrm{X}_{2}$ and $\mathrm{X}_{3}$ are independent and follow zero-mean symmetric distributions, the modified \gls*{blue}-based estimator gives the same estimates of $x_{2}$ and $x_{3}$ as \eqref{eq:Linear_est}. Therefore, the phase estimation is not affected as the optimal phase shift is a function of $x_2$ and $x_3$.
\end{rem}

\subsection{Empirical RMSE of the Estimators}\label{RMSE}
In the previous section, the proposed \gls*{ml} estimator in \eqref{eq:Exact_problem} was optimal for the estimation of $\bm{x}$, but, it was sub-optimal for estimating $\vartheta=\arg{\left(x_{2}+j\,x_{3}\right)}$. Additionally, the linear estimator was sub-optimal for estimating $\vartheta$  as it minimizes the \gls*{mse} with respect to $\bm{x}$, not $\vartheta$. 
In this section, we evaluate the \gls*{rmse} of the estimations of the parameter $\vartheta$ for linear and \gls*{ml} estimators using Monte Carlo simulations. We consider the set of measurement phases $\Phi=\left\{0, \frac{2\pi}{3}, \frac{4\pi}{3}\right\}$ for both estimators. For this numerical evaluation, we define the \gls*{snr} as
\begin{equation}
    \mathsf{SNR}\delequal \frac{\left|z_{0}\right|^{2}+\left|z\right|^{2}}{2\sigma^{2}}.
\end{equation}
Without loss of generality, we assume that $z_{0}=1$. Hence, according to \eqref{eq:opt_theta_3}, we have $\vartheta=-\arg{\left(z\right)}$. 
We explore different values for $|z|$, $\vartheta$ (i.e., $-\arg{\left(z\right)}$), and the \gls*{snr} to encompass different parameter scenarios.
To calculate the \gls*{rmse}, we constrain $\widehat{\vartheta}$ within the range $-\pi + \vartheta < \widehat{\vartheta} \leq \pi + \vartheta$.

As seen in Fig.~\ref{fig:RMSE}, for the considered $\Phi$, the \gls*{rmse} of the \gls*{ml} and linear estimators closely track each other and vary with the actual value of the parameter $\vartheta$. Additionally, we observe an increase in the respective \gls*{rmse} values for $|z|=1$, $|z|=3$, and $|z|=10$.

Thus, the linear and \gls*{ml} estimators perform similarly in terms of \gls*{rmse} for the phase estimation. Additionally, the \gls*{ml} estimator has two drawbacks: it requires knowledge of the noise variance and its solution involves complicated modified Bessel functions. Therefore, we consider the linear estimator to develop the multi-phase control algorithm.

\begin{rem}
    Since the proposed estimation techniques do not have any CSI or any a priori distribution of the channel values and only consider the conditional received power distribution, they are general regardless of the channel distributions. 
\end{rem}

 \begin{figure*}[!t]
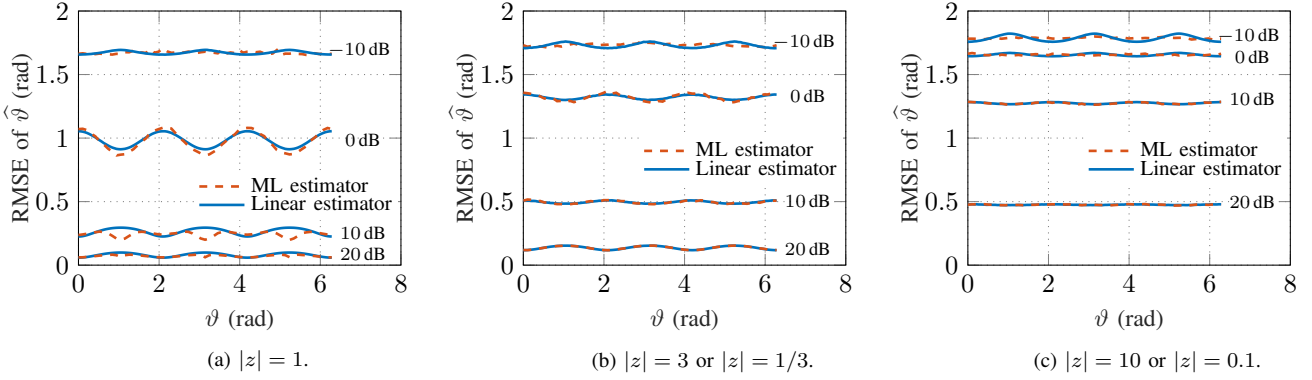

          \captionsetup[subfigure]{oneside,margin={2cm,0cm}}
          \hspace{-5mm}
          \subfloat[$\left|z\right|=1$.]{\input{Figures/RMSE_ML_vs_Linear_1}%
              \label{fig:RMSE_z1}}
          \hfil
          \hspace{-9mm}
          \subfloat[$\left|z\right|=3$ or $\left|z\right|=1/3$.]{\input{Figures/RMSE_ML_vs_Linear_2}%
             \label{fig:RMSE_z2}}
          \hfil
          \hspace{-9mm}
          \subfloat[$\left|z\right|=10$ or $\left|z\right|=0.1$.]{\input{Figures/RMSE_ML_vs_Linear_3}%
             \label{fig:RMSE_z3}}   
          \caption{The impact of the $\vartheta$ on the \gls*{rmse} of $\widehat{\vartheta}$ for \gls*{ml} and linear estimators. We consider $\Phi=\left\{0, \frac{2\pi}{3}, \frac{4\pi}{3}\right\}$ and various values for $|z|$ and  \gls*{snr}.}
          \label{fig:RMSE}
\end{figure*}

\section{Proposed Phase Alignment Schemes}\label{sec:PS}
In this section, we investigate the problem of maximizing the received \gls*{rf} power using a dynamic sequence of power measurements. We propose a model for the \gls*{ris} \gls*{eh} operations and we consider two different scenarios: 1. continuous phase control, where the \gls*{ris} elements can add any continuous phase shifts to the incident \gls*{rf} signal;  2. discrete phase control, where the \gls*{ris} elements can only add phase shifts from a predefined discrete set. For each scenario, we propose an algorithm to find the optimal phase of \gls*{ris} elements.

A simplified model of the operation of an \gls*{ris} with an \gls*{eh} module is shown in Fig.~\ref{fig:Harvesting_Structure}. The network entity manager can assess the network environment, including but not limited to the network demand, the \gls*{snr}, and power measurements from the \gls*{ris}, and determine if the \gls*{ris} should function to provide connectivity or to harvest energy. Within the \gls*{ris} controller, the preferred phase of each element can be determined based on its functionality, and the controller can adjust the phase of the elements accordingly. The energy harvester can be located either outside or inside the \gls*{ris} surface. It sends power measurements to a phase alignment algorithm implemented in the network entity manager module. These values can be used by the network entity manager to decide whether the \gls*{ris} should be in energy harvesting or data transmission mode. After reaching an appropriate phase configuration, the harvested energy can be fed back to the power supply to be used by the \gls*{ris}.

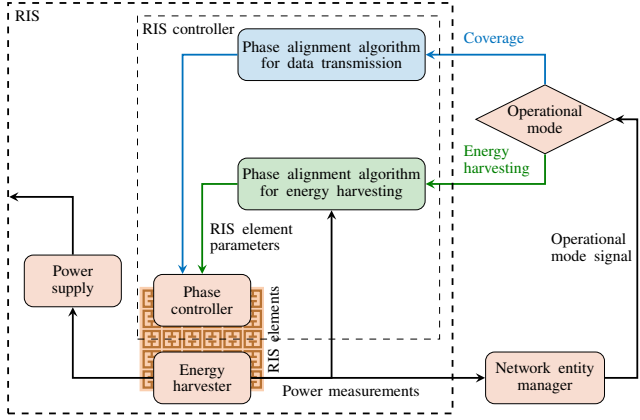
\begin{figure}[t]
    \hspace{-14mm}
    \scalebox{0.86}{
\definecolor{mycolor1}{rgb}{0.85,0.325,0.098}
\definecolor{mycolor2}{rgb}{0.00000,0.44700,0.74100}
\definecolor{mycolor3}{rgb}{0.00000,0.49804,0.00000}

\definecolor{apricot}{rgb}{0.98, 0.81, 0.69}
\definecolor{copper}{rgb}{0.72, 0.45, 0.2}
\definecolor{babyblue}{rgb}{0.54, 0.81, 0.94}
\definecolor{gray}{rgb}{0.5, 0.5, 0.5}
\usetikzlibrary{decorations.pathreplacing}
\linespread{0.6}
\tikzstyle{startstop} = [rectangle, rounded corners, minimum width=1.5cm, minimum height=0.8cm,text centered, draw=black, fill=mycolor1!20, inner sep=-0.3ex]
\tikzstyle{startstop_blue} = [rectangle, rounded corners, minimum width=1.5cm, minimum height=0.8cm,text centered, draw=black, fill=mycolor2!20, inner sep=-0.3ex]
\tikzstyle{startstop_green} = [rectangle, rounded corners, minimum width=1.5cm, minimum height=0.8cm,text centered, draw=black, fill=mycolor3!20, inner sep=-0.3ex]

\tikzstyle{decision} = [diamond, aspect=2,  text centered, draw=black, fill=mycolor1!20, inner sep=-0.3ex]
\tikzstyle{arrow} = [thick,->,>=stealth]

\begin{tikzpicture}[cross/.style={path picture={ 
  \draw[black]
(path picture bounding box.south east) -- (path picture bounding box.north west) (path picture bounding box.south west) -- (path picture bounding box.north east);
}}]



\node[text width=5.5em]  (NE) [startstop] {\scriptsize Network entity\\ manager};
\node[text width=4.5em] (OM) [decision, above of=NE, yshift = 3cm] {\scriptsize Operational mode};
\node[text width=8.5em] (PA1) [startstop_blue, left of=OM, yshift = 1cm, xshift = -2.3cm] {\scriptsize Phase alignment algorithm \\ for data transmission};
\node[text width=8.5em] (PA2) [startstop_green, left of=OM, yshift = -1cm, xshift = -2.3cm] {\scriptsize Phase alignment algorithm \\ for energy harvesting};

\draw[arrow] (NE.east) -- ++(0.5,0) |- coordinate[pos=0.25](m1)(OM.east);
\node[text width=4em]  [black,left=-0.2 of m1] {\scriptsize Operational mode signal};

\draw[arrow, color = mycolor2] (OM.north) |- coordinate[pos=0.65](m2)(PA1.east);
\node[text width=4em]  [black,above=0 of m2] {\scriptsize {\color{mycolor2} Coverage}};
\draw[arrow, color = mycolor3] (OM.south) |- coordinate[pos=0.65](m3)(PA2.east);
\node[text width=4em]  [black,above=0 of m3] {\scriptsize {\color{mycolor3} Energy harvesting}};

\pgfmathsetmacro{\R}{0.25}
\pgfmathsetmacro{\W}{1.25}
\pgfmathsetmacro{\H}{1.25}
\pgfmathsetmacro{\RISx}{-5.455}
\pgfmathsetmacro{\RISy}{0.6}

\filldraw[fill=apricot, draw=apricot] (\RISx-0.805,\RISy-0.8) rectangle (\RISx+1.1015,\RISy+0.8);
\foreach \i in {-2,...,3} {
    \foreach \j in {-2,...,2} {
    \node [draw, copper, line width=1,trapezium, trapezium left angle=90, trapezium right angle=90, minimum width=\R cm,outer sep=0pt,fill=apricot, rotate = 0] at (\RISx+\i*\W*\R,\RISy+\j*\H*\R) {};
    \draw[copper, line width=1] (\RISx+\i*\W*\R,\RISy+\j*\H*\R+0.5*\R) -- (\RISx+\i*\W*\R,\RISy+\j*\H*\R+0.15*\R);
    \draw[copper, line width=1] (\RISx+\i*\W*\R,\RISy+\j*\H*\R-0.5*\R) -- (\RISx+\i*\W*\R,\RISy+\j*\H*\R-0.15*\R);
    \draw[copper, line width=1] (\RISx+\i*\W*\R-0.25*\R,\RISy+\j*\H*\R-0.15*\R) -- (\RISx+\i*\W*\R+0.25*\R,\RISy+\j*\H*\R-0.15*\R);
    \draw[copper, line width=1] (\RISx+\i*\W*\R-0.25*\R,\RISy+\j*\H*\R+0.15*\R) -- (\RISx+\i*\W*\R+0.25*\R,\RISy+\j*\H*\R+0.15*\R);}}

\node[text width=4em] (PC) [startstop, left of=PA2, yshift = -1.8cm, xshift = -1cm] {\scriptsize Phase \\ controller};

\draw[arrow, color = mycolor3] (PA2.west) -| coordinate[pos=0.8](m4)(PC.north);
\node[text width=4em]  [black,right=0 of m4] {\scriptsize RIS element parameters};

\draw[arrow, color = mycolor2] (PA1.west) -- ++(-0.85,0) -- ++(0,-3.3935);

\node[text width=4em] (EH) [startstop, below of=PC, yshift = -0.2 cm] {\scriptsize Energy \\ harvester};

\draw[dashed] (-6.3,0.6) rectangle (-1.63,5.6);
\node at (-5.51, 5.4) {\scriptsize RIS controller};

\draw[arrow] (EH.east) -| coordinate[pos=0.5](m5)(PA2.south);
\node at (-3,-0.2) {\scriptsize Power measurements};

\draw[arrow] (EH.east) -- coordinate[pos=0.8](m6)(NE.west);

\node[text width=4em] (PS) [startstop, left of=PA2, xshift = -3 cm, yshift = -1.5 cm] {\scriptsize Power \\ supply};

\draw[arrow] (EH.west) -| coordinate[pos=0.8](m7)(PS.south);

\draw[arrow] (PS.north) |- (-8.284,2.8);

\draw[dashed, thick] (-8.3,-0.55) rectangle (-1.43,5.8);

\node at (-8, 5.6) {\scriptsize RIS};

\node at (-4.2, 0.75) {\scriptsize \rotatebox{90}{RIS elements}};

\end{tikzpicture}}
    \caption{ The abstract proposed model of the operation of a \gls*{ris} with an \gls*{eh} module.}
    \label{fig:Harvesting_Structure}
\end{figure}

\subsection{Continuous Phase Control}

\subsubsection{Noiseless Scenario}\label{sec:PS_1}
We begin by considering an ideal noiseless scheme, where the \gls*{ris} is capable of adding continuous phase shifts to the incident \gls*{em} waves. In the absence of the noise, the received powers become deterministic. Therefore, there is no need to apply any phase estimation technique. In Theorem \ref{Theorem_2}, we show that with three power measurements, we can determine the optimal phase shift without any error. First, we state Theorem \ref{Theorem_2}, and based on that, we develop our proposed algorithm that requires no a priori \gls*{csi}.

\begin{thm}\label{Theorem_2}
Let $f:\mathbb{R}\rightarrow \mathbb{R}$ be a function of the form
$$f{\left(\vartheta\right)}=\left|z_{0}+ ze^{j\vartheta}\right|^{2},$$ 
where $z_{0}, z\in \mathbb{C}$. Without knowing the explicit values of $z_{0}$ and $z$, the optimal variable $\vartheta^{\star}$  that maximizes $f{\left(\cdot\right)}$ can be computed based on the measurement vector $\bm{y}\delequal\left[f{\left(\varphi_{1}\right)}, f{\left(\varphi_{2}\right)}, f{\left(\varphi_{3}\right)}\right]^\mathsf{T}$ as $\vartheta^{\star}=\arg{\left(x_{2}+j x_{3}\right)}$,
where $\bm{x}\delequal \bm{A}^{-1}\bm{y}$. The matrix $\bm{A}$ is defined as 
\begin{equation} \label{eq:A}
\bm{A} \delequal \begin{bmatrix}
1 & \cos\left(\varphi_{1}\right) &  \sin\left(\varphi_{1}\right)\\
1 & \cos\left(\varphi_{2}\right) &  \sin\left(\varphi_{2}\right)\\
1 & \cos\left(\varphi_{3}\right) &  \sin\left(\varphi_{3}\right)
\end{bmatrix}.
\end{equation}
Note that $\varphi_{1}, \varphi_{2}, \varphi_{3}\in \left[0, 2\pi\right)$ must be selected such that $\det\!{\left(\bm{A}\right)}\neq 0$. 
\end{thm}
\begin{proof}\renewcommand{\qedsymbol}{}
The proof is provided in Appendix \ref{App_T2}.
\end{proof}

In general, Theorem \ref{Theorem_2} allows for the computation of the optimal phase shift using only three measurements, without requiring knowledge of the explicit expression of the function. 
Specifically, for the measurement phases $\varphi_{1}=0$, $\varphi_{2}=\pi/2$, and $\varphi_{3}=\pi$, a simple expression for the optimal phase shift can be obtained as follows:
\begin{equation}\label{eq:contin_phase_update}
\vartheta^{\star}=\arg{\left(y_{1}-y_{3}+j\left(2y_{2}-y_{1}-y_{3}\right)\right)}.
\end{equation}

\RestyleAlgo{ruled}
\SetKwComment{Comment}{/* }{ */}
\begin{algorithm}[t!] \label{Alg1}
\caption{The proposed phase-alignment algorithm for power maximization}
\textbf{Input:} The number of RIS elements $N$ and the number of iterations $M$\\
\textbf{Output:} Near-optimal phase vector $\bm{\vartheta}^{\star}$ that maximizes the received power.\\
\textbf{Initialize:} $\bm{\vartheta}\gets \bm{\vartheta}_{0}$, $m\gets 0$, and $\bm{e}_{n}$ is a vector, where the component $n$ is $1$ and all other components are $0$. 

\While{$m < M$}{
$m \gets m + 1$\;
\For{$n\gets 1$ \KwTo $N$}{
    $y_{1}\gets$ the measured power for the phase configuration $\bm{\vartheta}$\;
    $y_{2}\gets$ the measured power for the phase configuration $\bm{\vartheta}+\frac{\pi}{2}\bm{e}_{n}$\;
    $y_{3}\gets $ the measured power for the phase configuration $\bm{\vartheta}+\pi\bm{e}_{n}$\;
    $\bm{\vartheta} \gets \bm{\vartheta}+\arg\!\left(y_{1}-y_{3}+j\left(2y_{2}-y_{1}-y_{3}\right)\right)\bm{e}_{n}$\;
}}
$\bm{\vartheta}^{\star}\gets \bm{\vartheta}$\;
\end{algorithm}

Algorithm \ref{Alg1} is a sequential, iterative phase update algorithm. At each iteration, adjusting the phase of each element requires measuring the received power for three different phase configurations. This is the minimum number of power measurements that can be used for this purpose. The algorithm updates the phase of one element using \eqref{eq:contin_phase_update}, then proceeds to the next element until all $N$ elements have had their phases updated.   In this algorithm, the process is repeated $M$ times, but in practice, it can also be terminated earlier if some additional stopping criterion is satisfied. 

In the inner loop of Algorithm \ref{Alg1}, the number of operations remains constant regardless of the number of elements in the \gls*{ris} (i.e., $N$). Given that both the inner and outer loops iterate $N$ and $M$ times, respectively, the algorithm's complexity can be expressed as $\mathcal{O}(MN)$.

\begin{figure}
    \hspace{-15.6mm}
    \scalebox{0.83}{
\definecolor{mycolor1}{rgb}{0.85,0.325,0.098}
\definecolor{mycolor2}{rgb}{0.00000,0.44700,0.74100}
\definecolor{mycolor3}{rgb}{0.00000,0.49804,0.00000}

\definecolor{apricot}{rgb}{0.98, 0.81, 0.69}
\definecolor{copper}{rgb}{0.72, 0.45, 0.2}
\definecolor{babyblue}{rgb}{0.54, 0.81, 0.94}
\definecolor{gray}{rgb}{0.5, 0.5, 0.5}
\usetikzlibrary{decorations.pathreplacing}
\linespread{0.6}
\tikzstyle{startstop} = [rectangle, rounded corners, minimum width=1.5cm, minimum height=0.8cm,text centered, draw=black, fill=mycolor1!20, inner sep=-0.3ex]
\tikzstyle{startstop_blue} = [rectangle, rounded corners, minimum width=1.5cm, minimum height=0.8cm,text centered, draw=black, fill=mycolor2!20, inner sep=-0.3ex]
\tikzstyle{startstop_green} = [rectangle, rounded corners, minimum width=1.5cm, minimum height=0.8cm,text centered, draw=black, fill=mycolor3!20, inner sep=-0.3ex]

\tikzstyle{decision} = [diamond, aspect=2,  text centered, draw=black, fill=mycolor1!20, inner sep=-0.3ex]
\tikzstyle{arrow} = [thick,->,>=stealth]

\begin{tikzpicture}[cross/.style={path picture={ 
  \draw[black]
(path picture bounding box.south east) -- (path picture bounding box.north west) (path picture bounding box.south west) -- (path picture bounding box.north east);
}}]



\pgfmathsetmacro{\xshift}{4.1}
\pgfmathsetmacro{\xshiftt}{7.1}
\pgfmathsetmacro{\yshiftt}{0.75}

\pgfmathsetmacro{\xshifttt}{10}
\pgfmathsetmacro{\yshifttt}{0.5}

\draw[arrow, color = mycolor2, line width = 1.2] (0,0) -- (1,0.2);
\draw[arrow, color = mycolor2, line width = 1.2] (1,0.2) -- (0.9,0.9);
\draw[arrow, color = mycolor2, line width = 1.2] (0.9,0.9) -- (0.1,0.3);
\node[mycolor2] at (1, 0) {\scalebox{0.9}{$z_{1}$}};
\node[mycolor2] at (1.12, 0.85) {\scalebox{0.9}{$z_{2}$}};
\node[mycolor2] at (0.14, 0.58) {\scalebox{0.9}{$z_{3}$}};

\draw[arrow, color = mycolor2, dashed, line width = 1.2] (\xshift+0,0) -- (\xshift-1.0136,0.1126);
\draw[arrow, color = mycolor2, line width = 1.2] (\xshift-1.0136,0.1126) -- (\xshift-1.1136,0.8126);
\draw[arrow, color = mycolor2, line width = 1.2] (\xshift-1.1136,0.8126) -- (\xshift-1.9136,0.2126);
\draw[arrow, color = mycolor1, dotted, line width = 1.2] (\xshift-1.0136,0.1126) -- (\xshift-1.9136,0.2126);
\node[mycolor2] at (\xshift-0.58, 0.38) {\scalebox{0.9}{$z_{1}e^{j\vartheta^{1}_{1}}$}};
\node[mycolor2] at (\xshift-0.89, 0.78) {\scalebox{0.9}{$z_{2}$}};
\node[mycolor2] at (\xshift-2, 0.4) {\scalebox{0.9}{$z_{3}$}};
\node[mycolor1] at (\xshift-1.8, -0.03) {\scalebox{0.9}{$z_{2}{+}z_{3}$}};

\draw[arrow, color = mycolor2, dashed, line width = 1.2] (\xshiftt+0,\yshiftt+0) -- (\xshiftt-0.6829,\yshiftt- 0.1835);
\draw[arrow, color = mycolor2, line width = 1.2] (\xshiftt-0.6829,\yshiftt- 0.1835) -- (\xshiftt-1.6964,\yshiftt- 0.0709);
\draw[arrow, color = mycolor2, line width = 1.2] (\xshiftt-1.6964,\yshiftt- 0.0709) -- (\xshiftt-2.4964,\yshiftt- 0.6709);
\draw[arrow, color = mycolor1, dotted, line width = 1.2] (\xshiftt-0.6829,\yshiftt- 0.1835) -- (\xshiftt-2.4964,\yshiftt- 0.6709);
\node[mycolor2] at (\xshiftt-1.4, \yshiftt+0.18) {\scalebox{0.9}{$z_{1}e^{j\vartheta^{1}_{1}}$}};
\node[mycolor2] at (\xshiftt-0.42, \yshiftt+0.10) {\scalebox{0.9}{$z_{2}e^{j\vartheta^{1}_{2}}$}};
\node[mycolor2] at (\xshiftt-2.5, \yshiftt-0.39) {\scalebox{0.9}{$z_{3}$}};
\node[mycolor1] at (\xshiftt-2, \yshiftt-0.9) {\scalebox{0.9}{$z_{1}e^{j\vartheta^{1}_{1}}{+}z_{3}$}};

\draw[arrow, color = mycolor2, line width = 1.2] (\xshifttt+0,\yshifttt+0) -- (\xshifttt-0.6829,\yshifttt- 0.1835);
\draw[arrow, color = mycolor2, line width = 1.2] (\xshifttt-0.6829,\yshifttt- 0.1835) -- (\xshifttt-1.6964,\yshifttt- 0.0709);
\draw[arrow, color = mycolor2, dashed, line width = 1.2] (\xshifttt-1.6964,\yshifttt- 0.0709) -- (\xshifttt-2.6956,\yshifttt- 0.1127);
\draw[arrow, color = mycolor1, dotted, line width = 1.2] (\xshifttt+0,\yshifttt+0) -- (\xshifttt-1.6964,\yshifttt- 0.0709);
\node[mycolor2] at (\xshifttt-1.4, \yshifttt-0.36) {\scalebox{0.9}{$z_{1}e^{j\vartheta^{1}_{1}}$}};
\node[mycolor2] at (\xshifttt-0.42, \yshifttt-0.30) {\scalebox{0.9}{$z_{2}e^{j\vartheta^{1}_{2}}$}};
\node[mycolor2] at (\xshifttt-2.52, \yshifttt-0.32) {\scalebox{0.9}{$z_{3}e^{j\vartheta^{1}_{3}}$}};
\node[mycolor1] at (\xshifttt-1.3, \yshifttt+0.22) {\scalebox{0.9}{$z_{1}e^{j\vartheta^{1}_{1}}{+}z_{2}e^{j\vartheta^{1}_{2}}$}};


\pgfmathsetmacro{\xshiftttt}{3}
\pgfmathsetmacro{\yshiftttt}{-3}

\draw[arrow, color = mycolor2, dashed, line width = 1.2] (\xshiftttt+0,\yshiftttt+0) -- (\xshiftttt-1.0108,\yshiftttt - 0.1354);
\draw[arrow, color = mycolor2, line width = 1.2] (\xshiftttt-1.0108,\yshiftttt - 0.1354) -- (\xshiftttt-1.6937,\yshiftttt- 0.3189);
\draw[arrow, color = mycolor2, line width = 1.2] (\xshiftttt-1.6937,\yshiftttt- 0.3189) -- (\xshiftttt-2.6928,\yshiftttt- 0.3606);
\draw[arrow, color = mycolor1, dotted, line width = 1.2] (\xshiftttt-1.0108,\yshiftttt - 0.1354) -- (\xshiftttt-2.6928,\yshiftttt- 0.3606);
\node[mycolor2] at (\xshiftttt-0.6, \yshiftttt+0.25) {\scalebox{0.9}{$z_{1}e^{j\vartheta^{2}_{1}}$}};
\node[mycolor2] at (\xshiftttt-1.4, \yshiftttt-0.5) {\scalebox{0.9}{$z_{2}e^{j\vartheta^{1}_{2}}$}};
\node[mycolor2] at (\xshiftttt-2.52, \yshiftttt-0.6) {\scalebox{0.9}{$z_{3}e^{j\vartheta^{1}_{3}}$}};
\node[mycolor1] at (\xshiftttt-2.25, \yshiftttt-0.01) {\scalebox{0.9}{$z_{2}e^{j\vartheta^{1}_{2}}{+}z_{3}e^{j\vartheta^{1}_{3}}$}};

\pgfmathsetmacro{\xshifttttt}{6.5}
\pgfmathsetmacro{\yshifttttt}{-3.2}

\draw[arrow, color = mycolor2, dashed, line width = 1.2] (\xshifttttt+0,\yshifttttt+0) -- (\xshifttttt-0.7044,\yshifttttt - 0.0621);
\draw[arrow, color = mycolor2, line width = 1.2] (\xshifttttt-0.7044,\yshifttttt - 0.0621) -- (\xshifttttt-1.7152,\yshifttttt- 0.1975);
\draw[arrow, color = mycolor2, line width = 1.2] (\xshifttttt-1.7152,\yshifttttt- 0.1975) -- (\xshifttttt-2.7143,\yshifttttt- 0.2392);
\draw[arrow, color = mycolor1, dotted, line width = 1.2] (\xshifttttt-0.7044,\yshifttttt - 0.0621) -- (\xshifttttt-2.7143,\yshifttttt- 0.2392);
\node[mycolor2] at (\xshifttttt-1.4, \yshifttttt-0.4) {\scalebox{0.9}{$z_{1}e^{j\vartheta^{2}_{1}}$}};
\node[mycolor2] at (\xshifttttt-0.42, \yshifttttt+0.22) {\scalebox{0.9}{$z_{2}e^{j\vartheta^{2}_{2}}$}};
\node[mycolor2] at (\xshifttttt-2.52, \yshifttttt-0.5) {\scalebox{0.9}{$z_{3}e^{j\vartheta^{1}_{3}}$}};
\node[mycolor1] at (\xshifttttt-2.1, \yshifttttt+0.15) {\scalebox{0.9}{$z_{1}e^{j\vartheta^{2}_{1}}{+}z_{3}e^{j\vartheta^{1}_{3}}$}};

\pgfmathsetmacro{\xshiftttttt}{10}
\pgfmathsetmacro{\yshiftttttt}{-3.1}

\draw[arrow, color = mycolor2, line width = 1.2] (\xshiftttttt+0,\yshiftttttt+0) -- (\xshiftttttt-0.7044,\yshiftttttt - 0.0621);
\draw[arrow, color = mycolor2, line width = 1.2] (\xshiftttttt-0.7044,\yshiftttttt - 0.0621) -- (\xshiftttttt-1.7152,\yshiftttttt- 0.1975);
\draw[arrow, color = mycolor2, dashed, line width = 1.2] (\xshiftttttt-1.7152,\yshiftttttt- 0.1975) -- (\xshiftttttt - 2.7086,\yshiftttttt- 0.3118);
\draw[arrow, color = mycolor1, dotted, line width = 1.2] (\xshiftttttt+0,\yshiftttttt+0) -- (\xshiftttttt-1.7152,\yshiftttttt- 0.1975);
\node[mycolor2] at (\xshiftttttt-1.4, \yshiftttttt+0.15) {\scalebox{0.9}{$z_{1}e^{j\vartheta^{2}_{1}}$}};
\node[mycolor2] at (\xshiftttttt-0.42, \yshiftttttt+0.24) {\scalebox{0.9}{$z_{2}e^{j\vartheta^{2}_{2}}$}};
\node[mycolor2] at (\xshiftttttt-2.52, \yshiftttttt+0.05) {\scalebox{0.9}{$z_{3}e^{j\vartheta^{2}_{3}}$}};
\node[mycolor1] at (\xshiftttttt-1.2, \yshiftttttt-0.5) {\scalebox{0.9}{$z_{1}e^{j\vartheta^{2}_{1}}{+}z_{2}e^{j\vartheta^{2}_{2}}$}};

\node[black] at (0.6, -1) {\scalebox{0.9}{Initial Vectors}};
\node[black] at (3.1, -1) {\scalebox{0.9}{Iter 1: Step 1}};
\node[black] at (5.8, -1) {\scalebox{0.9}{Iter 1: Step 2}};
\node[black] at (8.6, -1) {\scalebox{0.9}{Iter 1: Step 3}};

\node[black] at (1.7, -4.6) {\scalebox{0.9}{Iter 2: Step 1}};
\node[black] at (5.2, -4.6) {\scalebox{0.9}{Iter 2: Step 2}};
\node[black] at (8.6, -4.6) {\scalebox{0.9}{Iter 2: Step 3}};

\end{tikzpicture}}
    \caption{Visualization of the proposed algorithm in different steps for a toy example with $N=3$.}
    \label{fig:Visualization_Alg}
\end{figure}
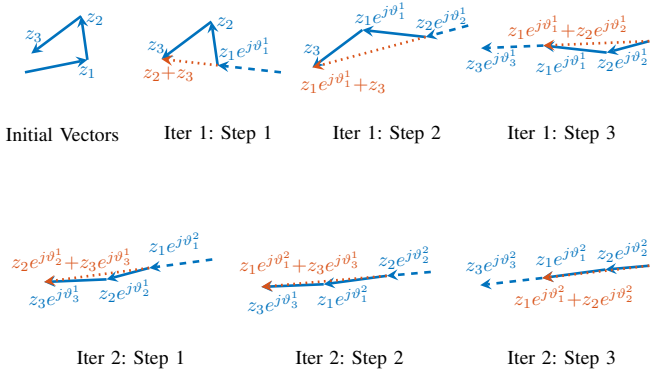

Fig.~\ref{fig:Visualization_Alg} presents a toy example that demonstrates the different steps of the proposed algorithm. The figure shows that, initially, the vectors are misaligned, leading to a relatively small amplitude of their sum compared to a scenario where they are aligned in the same direction. The algorithm rotates the phase of each vector to match the sum of the others. As the algorithm progresses and reaches the end of the second iteration, the vectors become almost aligned in the same direction, resulting in a nearly maximum amplitude of their sum.

\begin{thm}\label{Theorem_3}
The proposed Algorithm \ref{Alg1} converges to the maximum value of the function $f{\left(\bm{\vartheta}\right)}=\left|\sum_{n=1}^{N} z_{n}e^{j\vartheta_{n}}\right|^{2}$ as $M\to\infty$. 
\end{thm}
\begin{proof}\renewcommand{\qedsymbol}{}
The proof is provided in Appendix \ref{App_T3}.
\end{proof}

\subsubsection{Noisy Scenario}
In this section, we address the problem of phase alignment in noisy environments. We start by considering a single-phase control scenario and presenting estimators for the optimal phase shift. Subsequently, we extend this to the multiple-phase control scenario.

In Section \ref{sec:PS_1}, we developed Algorithm \ref{Alg1} to address the noiseless scenario by extending the single-phase control approach to the multiple-phase control scenario. Theorem \ref{Theorem_3} proves the convergence of Algorithm \ref{Alg1} to the optimal solution. In the noisy scenario, a similar extension can be made.

In this section, we propose a sub-optimal sequential algorithm that aims at maximizing the mean received power. For a phase configuration $\bm{\vartheta}=(\vartheta_{1}, \vartheta_{2},\dots, \vartheta_{N})$, the measured received power in the presence of the noise is
\begin{equation}\label{eq:Noisy_N}
    \mathrm{Y}= \left|\sum_{n=1}^{N}z_{n}e^{j\vartheta_{n}}+\mathrm{W}\right|^{2},
\end{equation}
where $\mathrm{W}\sim \mathcal{CN}{\left(0,\sigma^{2}\right)}$, and $\mathrm{Y}$ is the measured power and it is a random variable due to the noise. 
We define the average \gls*{snr} per element as
\begin{equation}\label{eq:SNR_def}
    \mathsf{SNR}\delequal \frac{\sum_{n=1}^{N}\mathbb{E}{\left(\left|\mathrm{Z}_{n}\right|^{2}\right)}}{N\sigma^{2}}.
\end{equation}

In Section \ref{Sec:Single_phase}, we considered the single-phase control scenario in the presence of noise. 
In the multi-phase control scenario with noise, we will make a similar extension as in Algorithm \ref{Alg1} to achieve Algorithm \ref{Alg2_Noisy}. In each iteration of this proposed algorithm, adjusting a single phase requires $L$ measurements. The algorithm updates the phase of one element using \eqref{eq:Linear_est} and \eqref{eq:est_vartheta}, then proceeds to the next element until all $N$ elements have had their phases updated. To adjust all the phases once, a total of $L\times N$ measurements are required. This process is repeated $M$ times. We further simplify the solution by using optimal measurement phases in Algorithm \ref{Alg3_Noisy}, where we apply \eqref{eq:noiseless_theta_hat} for the phase update of each element. 

Regarding the complexity of Algorithm \ref{Alg2_Noisy} and Algorithm \ref{Alg3_Noisy}, the number of operations within the inner loop scales linearly with $L$. Furthermore, with the inner and outer loops iterating $N$ and $M$ times, respectively, the algorithm's complexity can be represented as $\mathcal{O}(MNL)$.

Algorithm \ref{Alg3_Noisy} is computationally more efficient than Algorithm \ref{Alg2_Noisy}, as it performs the vector multiplication $\bm{y}^{\mathsf{T}}\bm{b}$ instead of the matrix multiplication $\bm{A}^{\dagger}\bm{y}$.

\begin{algorithm}[t!] \label{Alg2_Noisy}
\caption{Linear phase-alignment algorithm for power maximization in the presence of noise}
\textbf{Input:} The number of RIS elements $N$, the number of iterations is $M$, and set of measurement phases $\Phi=\left\{\varphi_{1}, \varphi_{2}, \dots, \varphi_{L}\right\}$.\\
\textbf{Output:} Near-optimal phase vector $\bm{\vartheta}^{\star}$ that maximizes the mean received power.\\
\textbf{Initialize:} $\bm{\vartheta}\gets \bm{\vartheta}_{0}$, $m\gets 0$, and $\bm{e}_{n}$ is a vector, where the component $n$ is $1$ and all other components are $0$. \\
\For{$l\gets 1$ \KwTo $L$}{
    $\bm{a}_{l}\gets \left[1, \cos\left(\varphi_{l}\right), \sin\left(\varphi_{l}\right)\right]^{\mathsf{T}}$\;
}
$\bm{A}\gets [\bm{a}_{1}, \dots, \bm{a}_{L}]^{\mathsf{T}}$\;
\While{$m < M$}{
$m \gets m + 1$\;
\For{$n\gets 1$ \KwTo $N$}{
    $\bm{y}\gets \left[0, \dots, 0\right]^{\mathsf{T}}$\;
    \For{$l\gets 1$ \KwTo $L$}{
    $y_{l}\gets$ the measured power for the phase configuration $\bm{\vartheta}+\varphi_{l}\bm{e}_{n}$\;
    }
$\widehat{\bm{x}}\gets \bm{A}^{\dagger}\bm{y}$\; 
    $\bm{\vartheta} \gets \bm{\vartheta}+\arg{\left(\widehat{x}_{2}+j\widehat{x}_{3}\right)}\bm{e}_{n}$\;
}}
$\bm{\vartheta}^{\star}\gets \bm{\vartheta}$\;
\end{algorithm}

\begin{algorithm}[ht] \label{Alg3_Noisy}
\caption{Optimal linear phase-alignment algorithm for power maximization in the presence of noise}
\textbf{Input:} The number of RIS elements $N$, the number of iterations $M$, and the number of measurement phases $L$.\\
\textbf{Output:} Near-optimal phase vector $\bm{\vartheta}^{\star}$ that maximizes the mean received power.\\
\textbf{Initialize:} $\bm{\vartheta}\gets \bm{\vartheta}_{0}$, $m\gets 0$, and $\bm{e}_{n}$ is a vector, where the component $n$ is $1$ and all other components are $0$. \\
$\bm{b}\gets \left[1, e^{j2\pi/L}, \dots, e^{j2\pi\left(L-1\right)/L}\right]^{\mathsf{T}}$\;
\While{$m < M$}{
$m \gets m + 1$\;
\For{$n\gets 1$ \KwTo $N$}{
    $\bm{y}\gets \left[0, \dots, 0\right]^{\mathsf{T}}$\;
    \For{$l\gets 1$ \KwTo $L$}{
    $y_{l}\gets$ the measured power for the phase configuration $\bm{\vartheta}+\frac{2\pi{\left(l-1\right)}}{L}\bm{e}_{n}$\;
    }
    $\bm{\vartheta} \gets \bm{\vartheta}+\arg{\left(\bm{y}^{\mathsf{T}}\bm{b}\right)}\bm{e}_{n}$\;
}}
$\bm{\vartheta}^{\star}\gets \bm{\vartheta}$\;
\end{algorithm}
\subsection{Discrete Phase Control}\label{sec:Discrete_Scenario}
In practice, with the current technology, a \gls*{ris} element can only apply a phase shift that takes a value from a finite set of possible phase shifts $\Omega$. We define this set as
\begin{equation*}
    \Omega \delequal\left\{\omega_{1}, \omega_{2}, \dots, \omega_{\left|\Omega\right|}\right\},
\end{equation*}
where $\omega_{k}\in[0,2\pi)$ for all $1\leq k\leq \left|\Omega\right|$, and $\omega_{k}\neq \omega_{k'}$ for all $1\leq k,k'\leq \left|\Omega\right|$ and $k\neq k'$. 

The problem of finding the optimal phase-shifts for the \gls*{ris} elements that maximize the received power becomes
\begin{maxi}[2]
{\scriptstyle \bm{\vartheta}}{f{\left(\bm{\vartheta}\right)}}{\hspace{-10.4mm}}{\label{eq:Prob_DS}} 
\addConstraint{\bm{\vartheta}\in \Omega^{N}},
\end{maxi}
where $f{\left(\bm{\vartheta}\right)}\delequal\left|\sum_{n=1}^{N} z_{n}e^{j\vartheta_{n}}\right|^{2}$. In general, the problem \eqref{eq:Prob_DS} might have more than one solution. For instance, when the set $\Omega$ is closed under modular $2\pi$ addition, assuming $\bm{\vartheta}^{\star}$ is a solution, one can easily show that all the members of the set $\left\{\bm{\vartheta}^{\star}+\bm{1}\omega_{k}\mod 2\pi : 1\leq k\leq \left|\Omega\right|\right\}$ are also solutions.

Below, we introduce a modified discrete version of the Theorem \ref{Theorem_2}.

\begin{thm}\label{Theorem_4}
Let $f:\Omega\rightarrow \mathbb{R}$ be a function with the expression
$$f{\left(\vartheta\right)}=\left|z_{0}+ ze^{j\vartheta}\right|^{2},$$ 
where $z_{0}, z\in \mathbb{C}$ are explicitly unknown. For $\left|\Omega\right|\geq 3$, the input variable $\vartheta=\omega_{k^{\star}}$ maximizes $f{\left(\cdot\right)}$, where
\begin{equation}
      k^{\star}\delequal\underset{1\leq k\leq \left|\Omega\right|}{\mathrm{argmin}}\  \min{\left(\zeta_{k}, 2\pi-\zeta_{k}\right)},
\end{equation}
$\zeta_{k} \!\delequal\! \vartheta^{\star}\!-\omega_{k} \!\!\mod 2\pi\,$ for all $1\!\leq\! k\!\leq\!\left|\Omega\right|$, and $\vartheta^{\star}$ is the optimal continuous phase shift that is computed using Theorem \ref{Theorem_2}. 
\end{thm}
\begin{proof}\renewcommand{\qedsymbol}{}
The proof is provided in Appendix \ref{App_T4}.
\end{proof}
We utilize Theorem \ref{Theorem_4} to propose the Algorithm \ref{Alg4} for the discrete scenario, which is a modified version of Algorithm \ref{Alg1}. The new Algorithm \ref{Alg4} is a natural extension of Algorithm \ref{Alg1} for the discrete phase-shift scenario, but it is not necessarily optimal.

The number of operations within the inner loop of Algorithm \ref{Alg4} scales linearly with $|\Omega|$. Moreover, with the inner and outer loops iterating $N$ and $M$ times, respectively, the algorithm's complexity can be expressed as $\mathcal{O}(MN|\Omega|)$.

\begin{thm}\label{Theorem_5}
The proposed Algorithm \ref{Alg4} converges as $M\to\infty$.
\end{thm}
\begin{proof}\renewcommand{\qedsymbol}{}
The proof is provided in Appendix \ref{App_T5}.
\end{proof}

\begin{algorithm}[t] \label{Alg4}
\caption{The proposed discrete phase-alignment algorithm for power maximization}
\textbf{Input:} The number of RIS elements $N$, the number of iterations $M$, and the set of possible phase shifts $\Omega=\left\{\omega_{1}, \omega_{2}, \dots, \omega_{\left|\Omega\right|}\right\}$. \\
\textbf{Output:} Near-optimal phase vector $\bm{\vartheta}^{\star}\in\Omega^{N}$ that maximizes the received power.\\
\textbf{Initialize:} $\bm{\vartheta}\gets \bm{\vartheta}_{0}\in\Omega^{N}$, $m\gets 0$, and $\bm{e}_{n}$ is an $N$-dimensional vector, where the component $n$ is $1$ and all other components are $0$. Select three measurement phase shifts $\varphi_{1}, \varphi_{2}, \varphi_{3}\in\Omega$, such that $\sin{\left(\varphi_{1}-\varphi_{3}\right)}+\sin{\left(\varphi_{2}-\varphi_{1}\right)}+\sin{\left(\varphi_{3}-\varphi_{2}\right)}\neq 0$. 
\\
    $\bm{A} \gets \begin{bmatrix}
    1 & \cos\left(\varphi_{1}\right) &  \sin\left(\varphi_{1}\right)\\
    1 & \cos\left(\varphi_{2}\right) &   \sin\left(\varphi_{2}\right)\\
    1 & \cos\left(\varphi_{3}\right) &  \sin\left(\varphi_{3}\right)
    \end{bmatrix}$;

\While{$m < M$}{
$m \gets m + 1$\;
$\bm{\vartheta}^{\textup{old}}\gets \bm{\vartheta}$\;
\For{$n\gets 1$ \KwTo $N$}{
    $y_{1}\gets$ the measured power for the phase configuration $\bm{\vartheta}+\left(\varphi_{1}-\vartheta_{n}\right)\bm{e}_{n}$\;
    $y_{2}\gets$ the measured power for the phase configuration $\bm{\vartheta}+\left(\varphi_{2}-\vartheta_{n}\right)\bm{e}_{n}$\;
    $y_{3}\gets$ the measured power for the phase configuration $\bm{\vartheta}+\left(\varphi_{3}-\vartheta_{n}\right)\bm{e}_{n}$\;
    $\bm{x}\gets \bm{A}^{-1}\bm{y}$\;
    $\alpha \gets \arg{\left(x_{2}+jx_{3}\right)}+\vartheta_{n}$\;
    \For{$k\gets 1$ \KwTo $\left|\Omega\right|$}{
    $\zeta_{k}\gets \alpha-\omega_{k}\mod 2\pi$\;
    }
    $k^{\star}\gets\underset{1\leq k\leq \left|\Omega\right|}{\mathrm{argmin}}\  \min{\left(\zeta_{k}, 2\pi-\zeta_{k}\right)}$\;

    $\bm{\vartheta} \gets \bm{\vartheta}+\left(\omega_{k^{\star}}-\vartheta_{n}\right)\bm{e}_{n}$\;
}
\If{$\bm{\vartheta}^{\textup{old}}=\bm{\vartheta}$}{
      \text{break}\;
   }
}
$\bm{\vartheta}^{\star}\gets \bm{\vartheta}$\;
\end{algorithm}

\section{Performance Evaluation}\label{sec:NumResults}
In this section, we evaluate the performance of our proposed algorithm and compare it to a random phase update method using Monte-Carlo simulations. We consider a \gls*{ris} with $100$ elements and generate random complex Gaussian distributed values with unit variance for $\left\{\mathrm{Z}_n\right\}_{n=1}^N$.
\subsection{Benchmark: Random Algorithm}
Since there is no prior work to compare against, we use a random algorithm as the benchmark for the proposed algorithm. The basic concept is that, at each step, the algorithm picks a single element of the \gls*{ris} sequentially and assigns a random phase value from a uniform distribution over the interval $[0, 2\pi)$ in the continuous scenario or over the set $\Omega$ in the discrete scenario. In the discrete scenario, we ensure that the new random phase differs from the previously used one. The new power measurement is then compared to the previous one. If the power increases, the algorithm updates the phase and moves on to the next element. If the measured power decreases, the algorithm maintains the previous phase and proceeds to the next element.

\subsection{Performance Evaluation}

We define the \gls*{nap} as the actually received power divided by the theoretically maximum received power level. 
More formally, the \gls*{nap} is defined as
\begin{equation}\label{eq:Norm_Acheived_Pow}
    \mathsf{NAP} \delequal \cfrac{\left|\sum_{n=1}^{N}\mathrm{Z}_{n}e^{j\vartheta_{n}}\right|^{2}}{\left(\sum_{n=1}^{N}\left|\mathrm{Z}_{n}\right|\right)^{2}},
\end{equation}
where $\bm{\vartheta}$ is the used \gls*{ris} phase configuration for different methods. 
Furthermore, the \gls*{mnap} is defined as $\mathsf{MNAP} \delequal \mathbb{E}{\left(\mathsf{NAP}\right)}$.

Fig.~\ref{fig:Mean_vs_measurements} illustrates the \gls*{mnap} versus the number of power measurements for the proposed Algorithm \ref{Alg3_Noisy} and random benchmark method under different \gls*{snr} conditions. The results are applicable to both indirect and direct \gls*{eh}, but we stress that the former typically leads to lower SNRs due to the extra pathloss.
The proposed method uses three measurements per phase update ($L=3$) of each element and converges to its final value after $300$ measurements (i.e., $LN$), while the random one requires ten times more measurements as it has a slow convergence rate. We notice that  the proposed method converges equally fast at any \gls*{snr}, while the convergence rate decreases for the random method when the \gls*{snr} increases. Importantly, the proposed method converges to a solution with a higher \gls*{rf} power than the random one. The \gls*{mnap} increases with the \gls*{snr}. Therefore, with three measurements per phase update, it can't reach its full potential in terms of receiving power from the ambient \gls*{rf} source. At low \glspl*{snr}, the proposed method does not reach the maximum theoretical power. This happens because the \gls*{mse} of the phase estimation does not decrease at each step, preventing the algorithm from converging to the theoretically optimal phase configuration. To overcome this issue, we can increase the number of measurements per phase update ($L$) in order to reduce the \gls*{mse} of the phase estimation. The \gls*{mnap} of the proposed method at $\mathsf{SNR}\to\infty$ can be regarded as a genie-aided phase update within the framework of our sequential algorithm. Moreover, with access to \gls*{csi}, we can achieve $\mathsf{MNAP}=1$. Thus, it can be considered as the upper bound with genie-aided assistance. 

\begin{figure}
    \centering
    \input{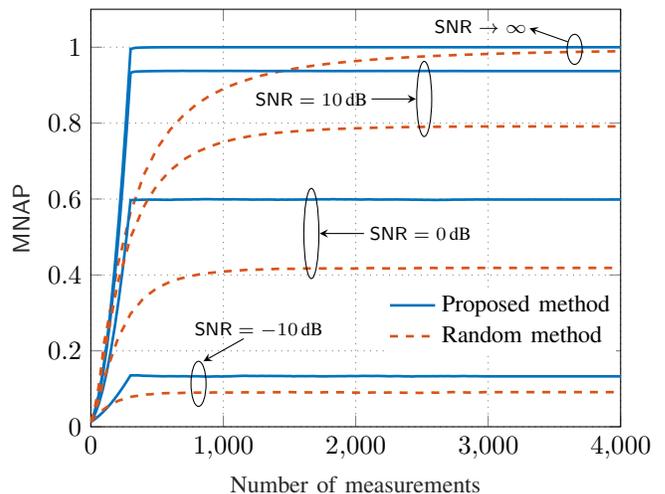}
    \caption{Mean normalized achieved power versus the number of measurements for the proposed Algorithm \ref{Alg3_Noisy} ($L=3$) and random method, and for various \glspl*{snr}.}
    \label{fig:Mean_vs_measurements}
\end{figure}
\begin{figure}
    \centering
    \input{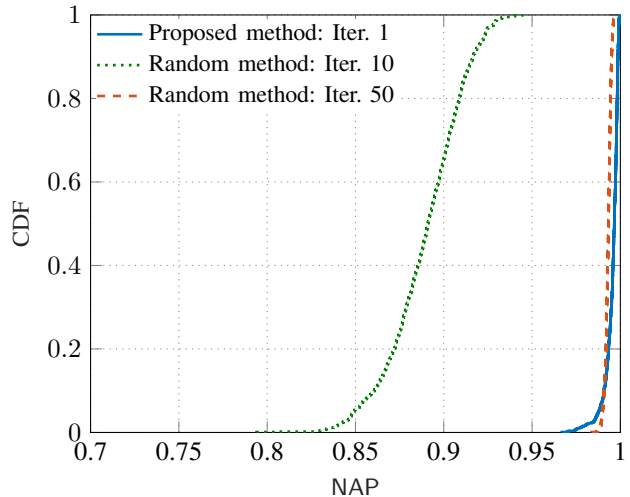}
    \caption{The CDF of the normalized achieved power for the proposed Algorithm \ref{Alg1} and random method in noiseless scenario.}
    \label{fig:CDF_proposed}
\end{figure}

Fig.~\ref{fig:CDF_proposed} shows the \gls*{cdf} of the relative achieved power for both the proposed and random algorithms, compared to the optimum in noiseless scenario. The randomness in the proposed method is due to the random channel realization, while in the random method it is due to both random channel realizations and random phase updates. We conducted $1000$ simulations with randomly generated channels. The results indicate that one iteration of the proposed algorithm achieves significantly better performance than the random method achieves after ten iterations. Although the proposed algorithm still performs slightly better than the random one after fifty iterations (e.g., on the average, as also shown in Fig.~\ref{fig:Mean_vs_measurements}), the latter provides a CDF curve with smaller variations.

In Fig.~\ref{fig:Max_vs_SNR}, the impact of the number of measurements per phase update ($L$) on the \gls*{mnap} of the proposed Algorithm \ref{Alg3_Noisy} and the random method is shown. It can be observed that as $L$ increases, the performance improves and the \gls*{ris} achieves higher power. For example, at $\mathsf{SNR}=-10$\,dB, the \gls*{mnap} is $0.09$ for the random method, while for the proposed method with $L=3$, $L=10$, $L=30$, and $L=100$, the \gls*{mnap}s are $0.14$, $0.31$, $0.54$, $0.75$, respectively. The convergence to the final value of the proposed algorithm occurs after almost $NL$ measurements. This increase in the number of measurements leads to an increase in the \gls*{nap} at low \glspl*{snr} but also increases the latency. Hence, there is a trade-off between the \gls*{nap} and the number of measurements. For instance, at $\mathsf{SNR}=10$\,dB, the \gls*{nap} for the proposed method with $L=3$ is $0.94$. In this case, additional measurements may not be necessary as they only slightly improve the \gls*{nap}. For low \glspl*{snr} scenarios, if the limit is $1000$ measurements, for a \gls*{ris} with $N=100$ elements, choosing $L=10$ is more advantageous compared to $L=3$, as it achieves more power.

It is worth mentioning that \gls*{ris} is envisioned to be deployed in environments where the channel between the transmitter and the receiver is weak. However, we deploy \gls*{ris} in a place where the channel between the transmitter and the \gls*{ris} and the channel between \gls*{ris} and the receiver are both strong. Hence, the low \gls*{snr} from the transmitter to the receiver does not imply the low \gls*{snr} from the transmitter to the \gls*{ris}. On the contrary, a  \gls*{ris} might only be effective if it has a relatively strong channel to both the transmitter and receiver \cite{9721205}. Thus, the \gls*{snr} can be sufficiently high to configure the \gls*{ris} as proposed in this paper.

\begin{figure}
    \centering
    \input{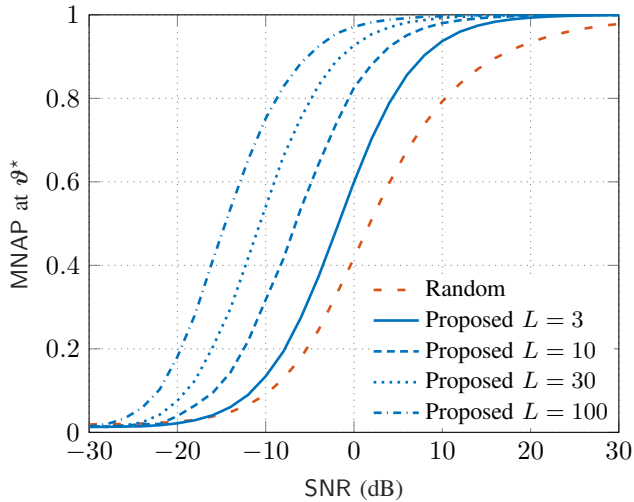}
    \caption{Impact of the \gls*{snr} on the mean normalized achieved power at $\bm{\vartheta}^{\star}$ of the proposed Algorithm \ref{Alg3_Noisy} and the random method.}
    \label{fig:Max_vs_SNR}
\end{figure}

In Fig.~\ref{fig:Discrete_Optimal}, we evaluate the performance of the proposed Algorithm \ref{Alg4} in the absence of the noise, considering a discrete set of possible phase shifts $\Omega=\left\{0, \pi/2, \pi, 3\pi/2\right\}$. The number of \gls*{ris} elements is $N=10$. We compare the \gls*{nap} 
of the proposed method to that of the random phase update method and the maximum feasible value (calculated using a brute-force method).  The error bars indicate a $95$ percent confidence interval. At the beginning, the random method achieves higher power, due to having one phase update per measurement, compared to the proposed method's one phase update per three measurements ($L=3$). However, as the \gls*{ris} takes more measurements, the proposed method outperforms the random one. The proposed method converges to a higher value than the random algorithm, but not exactly to the maximum feasible value. Therefore, our proposed method is sub-optimal for the discrete, noiseless scenario. We did not consider a larger number of \gls*{ris} elements as the number of possible phase configurations to find the maximum feasible value grows exponentially with $N$ (i.e., $\left|\Omega\right|^{N}$).

\begin{figure}
    \centering
    \input{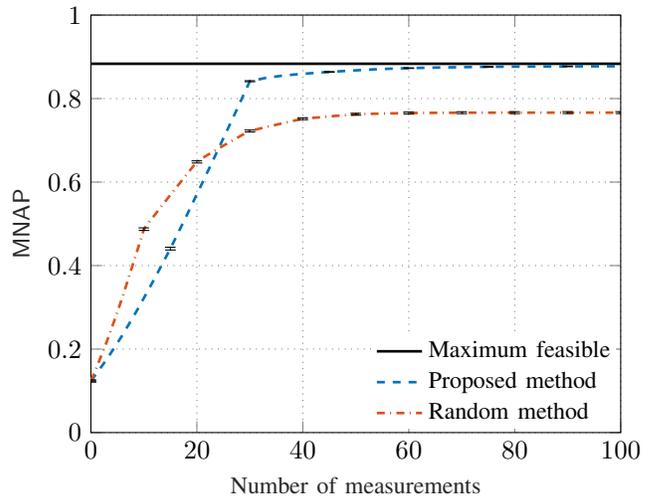}
    \caption{Mean normalized achieved power for the proposed Algorithm \ref{Alg4} and random method versus the number of measurements. The number of \gls*{ris} elements is $N=10$, and the set of possible phase shifts is $\Omega=\left\{0, \pi/2, \pi, 3\pi/2\right\}$. The maximum feasible value for the discrete scheme is shown with the black solid line and the error bars represent the $95$ percent confidence interval.}
    \label{fig:Discrete_Optimal}
\end{figure}

 To get a sense of the amount of power a device can harvest using \gls*{ris}, let's consider a scenario where an \gls*{ris} surface is situated on the $xy$ plane, with its center positioned at the origin. The elements form a grid with resolution $\lambda/2$, where we consider $\lambda=0.125$\,m (In this setup, a \gls*{ris} with $256$ elements occupies an area of $1$\,m$^{2}$).  Furthermore, we consider an isotropic \gls*{rf} source with a transmit power of $1$\,W at $[0,-3,4]^{\mathsf{T}}$ and an \gls*{eh} device at $[0,1,2]^{\mathsf{T}}$ (i.e., an example of indirect \gls*{eh}). We compute the channel gains between the \gls*{ris} elements and the transmit and receive antennas by taking into account the near-field and polarization effects \cite{9184098}.

We assume the following sigmoidal function for the \gls*{eh} conversion efficiency.
\begin{equation}\label{eq:Conversion_efficiency}
    \eta{\left(x\right)}\delequal\cfrac{P_{\text{sat}}\cdot\left(\left(1+e^{-a\left(x-b\right)}\right)^{-1}-\left(1+e^{ab}\right)^{-1}\right)}{x\cdot\left(1-\left(1+e^{ab}\right)^{-1}\right)},
\end{equation}
 where $P_{\text{sat}}$ implies the output harvesting saturation power. Constants $a$ and $b$ are associated with detailed circuit specifications, including factors like resistance, capacitance, and diode turn-on voltage. We adopt $a = 30$, $b = 0.07$, and $P_{\text{sat}}=0.1$\,W, and the unit of the input power $x$ is Watt \cite{7264986}.

Fig.~\ref{fig:Harvested_vs_N_Elements} illustrates the mean harvested power versus the number of \gls*{ris} elements for the proposed Algorithm \ref{Alg3_Noisy} and the random method (both after reaching convergence). We also considered a genie-aided scheme with available \gls*{csi}. As observed, increasing the number of \gls*{ris} elements (i.e., expanding the surface area) leads to an increase in harvested power until it reaches saturation. The saturation power, which is less than the saturation power of the harvesting device (i.e., $20$\,dBm), occurs due to the limited power that an infinitely large \gls*{ris} surface can reflect toward a receiver \cite{9184098}. Alternatively, if the transmit power were sufficiently large, saturation could result from the harvesting device.

Furthermore, an optimality gap in mean harvested power is observed between the genie-aided method and the proposed and random methods, and this gap becomes smaller at higher \glspl*{snr}.

The gap in mean harvested power between the proposed method and the random method demonstrates the superiority of the proposed method. For instance, for a \gls*{ris} with $1024$ elements (i.e., a surface $4$\,m$^{2}$), the gap between the proposed method and random method at the \gls*{snr} of $-20$\,dB is $2.27$\,dB, which is significant. It is worth mentioning that the proposed method converges with less number measurements than the random method (as shown previously in Fig.~\ref{fig:Mean_vs_measurements}).

\begin{figure}
    \centering
    \input{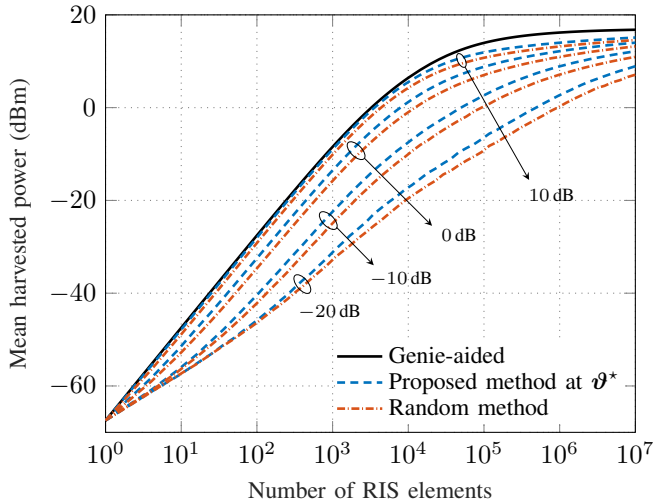}
    \caption{Impact of the number of the \gls*{ris} elements and \gls*{snr} on the harvested power of the proposed Algorithm \ref{Alg3_Noisy} and the random method (both after reaching convergence).}
    \label{fig:Harvested_vs_N_Elements}
\end{figure}

\section{Conclusions}\label{sec:Conclusion}
This paper has presented a method for energy harvesting at a \gls*{ris} from an ambient \gls*{rf} source in the absence of coordination with the source. The objective was to maximize the received power by adjusting the phases of the \gls*{ris} elements based only on power measurements, obtained without having a \gls*{rf} receiver. The proposed sequential algorithm outperformed the random phase update method in terms of achieved power after convergence, while requiring fewer measurements. Also, it is proved to converge to the optimum for the noiseless continuous phase scenario. The impact of the number of measurements per phase update on the achieved power was evaluated, showing that increasing the number of measurements leads to improved performance, although the amount of improvement depends on the \gls*{snr} of the system. A higher number of measurements leads to significant improvement in received power in low \gls*{snr} regimes but is unnecessary in high \gls*{snr} regimes where three measurements provide a near-optimal solution. Besides, \gls*{eh} is beneficiary when the signal is strong and the \gls*{snr} is large. 
For a discrete, noiseless scenario, the proposed method was found to be near-optimal, achieving a higher power than the random algorithm but not exactly the maximum feasible power. In future work, we plan to extend this study to consider more general scenarios.
For example, the proposed method can be used to estimate the channel to the transmitter and thereby enable RIS configuration for data transmission without the involvement of the BS.

\appendices
\section{Proof of Theorem \ref{Theorem_1}} \label{App_T1}
The matrix $\bm{A}^{\mathsf{T}}\bm{A}$ for a given measurement phase vector $\bm{\varphi}$ is
\begin{equation} \label{eq:ATA}
    \bm{A}^{\mathsf{T}}\!\bm{A} \!\delequal\! L\!\begin{bmatrix}
    1 & r_{1}\!\cos{\!\left(\delta_{1}\right)} &  r_{1}\!\sin{\!\left(\delta_{1}\right)}\\
    r_{1}\!\cos{\!\left(\delta_{1}\right)} & \frac{1\!+\!r_{2}\!\cos{\!\left(\delta_{2}\right)}}{2} &  \frac{r_{2}\!\sin{\!\left(\delta_{2}\right)}}{2} \\
    r_{1}\!\sin{\!\left(\delta_{1}\right)} & \frac{r_{2}\!\sin{\!\left(\delta_{2}\right)}}{2}  &  \frac{1\!-\!r_{2}\!\cos{\!\left(\delta_{2}\right)}}{2}
    \end{bmatrix},
\end{equation}
where $ r_{1}e^{j\delta_{1}}\delequal\sum_{l=1}^{L}e^{j\varphi_{l}}/L$ and  $ r_{2}e^{j\delta_{2}}\delequal\sum_{l=1}^{L}e^{j2\varphi_{l}}/L$. The equation for the eigenvalues of $\bm{A}^{\mathsf{T}}\bm{A}$ is $\det\!{\left(\bm{A}^{\mathsf{T}}\bm{A}-\lambda\bm{I}_{3}\right)}=0$,
where we know $\lambda_{i}=d_{i}^{2}$ for $1\leq i\leq 3$. After some mathematical manipulations, we get 
\begin{align}\label{eq:rho}
    \rho&\delequal\sum_{i=1}^{3}\frac{1}{d_{i}^{2}}=\frac{5-4r_{1}^2-r_{2}^2}{L\!\left(1\!-\!2r_{1}^{2}\!-r_{2}^2+\!2r_{1}^2r_{2}\cos{\left(\delta_{2}\!-\!2\delta_{1}\right)}\right)},
\end{align}
and since $\rho\geq 0$ and $0\leq r_{1}, r_{2}\leq 1$, we should have 
\begin{equation}\label{eq:theory_5_a}
    1\!-\!2r_{1}^{2}\!-r_{2}^2+\!2r_{1}^2r_{2}\!\cos{\left(\delta_{2}\!-\!2\delta_{1}\right)}>0.
\end{equation}

From $0\leq r_{1}, r_{2}\leq 1$ and Lemma \ref{lemma_2} with $\alpha = 0.4$, we have 
\begin{align}\label{eq:55}
     0.4 r_{2}^{2}+0.6r_{1}^{2}&\geq r_{1}^{2}r_{2}  &&\Longleftrightarrow \nonumber\\
     2 r_{2}^{2}+3 r_{1}^{2} - 5r_{1}^{2}r_{2}\cos{\left(\delta_{2}-2\delta_{1}\right)} &\geq 0  &&\stackrel{\eqref{eq:theory_5_a}}{\Longleftrightarrow}\nonumber\\
     \frac{5-4r_{1}^2-r_{2}^2}{\!1\!-\!2r_{1}^{2}\!-r_{2}^2+\!2r_{1}^2r_{2}\!\cos{\left(\delta_{2}\!-\!2\delta_{1}\right)}}&\geq 5  &&\stackrel{\eqref{eq:rho}}{\Longleftrightarrow}\nonumber\\
     \rho &\geq \frac{5}{L},
\end{align}
where the equality holds when $r_{1}=r_{2}=0$. According to Lemma \ref{lemma_3}, if $\bm{\varphi}^{\star}=(\varphi_{0}, \varphi_{0}\!+\!2\pi/L, \ldots,$ $\varphi_{0}\!+\!2\pi\left(L\!-\!1\right)/L)$, then $r_{1}=r_{2}=0$, making it an optimal solution that minimizes $\rho$.
Also, when $r_{1}=r_{2}=0$, we have $\bm{A}^{\mathsf{T}}\bm{A} =\operatorname{diag}\!{\left(\left[L, L/2, L/2\right]^{\mathsf{T}}\right)}$. Therefore, $d_{1}^{\star}=\sqrt{L}$, $d_{2}^{\star}=d_{3}^{\star}=\sqrt{L/2}$.
\qed
\begin{lem} \label{lemma_2}
Given $0\leq \alpha\leq 0.5$ and $M>0$, for any $0\leq x,y\leq M$, we have 
    \begin{equation}\label{eq:lemma_2}
         \alpha x^2+\left(1-\alpha\right)y^{2}\geq \frac{xy^{2}}{M}.
    \end{equation}
\end{lem}
\begin{proof}
For $x=0$ or $y=0$, \eqref{eq:lemma_2} holds. Lets consider $0<x,y\leq M$, assuming $g{\left(x\right)}\delequal \alpha x M^{-2}+\left(1-\alpha\right)x^{-1}$, we have $g''{\left(x\right)}\geq 0$. Hence, $g{\left(\cdot\right)}$ is a convex function. We can compute the optimum solution to minimize $g{\left(\cdot\right)}$ by solving $g'{\left(x^{\star}\right)}=0$. Hence, we get $x^{\star}=M\sqrt{\tfrac{1-\alpha}{\alpha}}$.
    From $0\leq \alpha\leq 0.5$, we have $x^{\star}\geq M$. Therefore, the minimum value of $g{\left(\cdot\right)}$ over the interval $(0, M]$ occurs at the boundary $x=M$, and we have $g{\left(M\right)}=1/M$.  
    Hence,
    \begin{align}
         \alpha x^2+\left(1-\alpha\right)y^{2}\geq xy^{2}g{\left(x\right)}\geq xy^{2}g{\left(M\right)}=\frac{xy^{2}}{M},
    \end{align}
    and the proof is complete.
\end{proof}

\begin{lem} \label{lemma_3}
Assuming $L\geq 3$, $\varphi_{0}\in\mathbb{R}$, and $\varphi_{l}=\varphi_{0}+2\pi\left(l-1\right)/L$ for all $1\leq l\leq L$, we have 
\begin{equation}
    \sum_{l=1}^{L}e^{j\varphi_{l}}=\sum_{l=1}^{L}e^{j2\varphi_{l}}=0.
\end{equation}
\end{lem}
\begin{proof}
For $L\geq 2$, we have $\sum_{l=1}^{L}e^{j\varphi_{l}}\!=\!e^{j\varphi_{0}}\!\sum_{l=1}^{L}e^{\frac{j2\pi\left(l-1\right)}{L}}
    \!=\!e^{j\varphi_{0}}\frac{1\!-\!e^{j2\pi\left(l\!-\!1\right)}}{1\!-\!e^{\frac{j2\pi}{L}}}=0$,
and for $L\geq 3$, we have $\sum_{l=1}^{L}e^{j2\varphi_{l}}\!=\!e^{j2\varphi_{0}}\!\sum_{l=1}^{L}e^{\frac{j4\pi\left(l-1\right)}{L}}
    \!=\!e^{j2\varphi_{0}}\frac{1\!-\!e^{j4\pi\left(l\!-\!1\right)}}{1\!-\!e^{\frac{j4\pi}{L}}}=0$,
and the proof is complete.
\end{proof}

\section{Proof of Theorem \ref{Theorem_2}} \label{App_T2}
Let us define $\bm{x}\delequal \left[\left|z_{0}\right|^{2}+\left|z\right|^{2},  2\operatorname{Re}\!{\left(z_{0}z^{*}\right)}, 2\operatorname{Im}\!{\left(z_{0}z^{*}\right)}\right]^{\mathsf{T}}$.
The variable $\vartheta^{\star}$ that maximizes $f{\left(\cdot\right)}$ is given by
\begin{align*}
     \vartheta^{\star}&=\arg{\left(z_{0}\right)}-\arg{\left(z\right)}=\arg{\left(z_{0}z^{*}\right)}\\
     &=\arg{\left(\operatorname{Re}\!{\left(z_{0}z^{*}\right)}+j\operatorname{Im}\!{\left(z_{0}z^{*}\right)}\right)}=\arg{\left(x_{2}+j x_{3}\right)}.
\end{align*}
Therefore, we can compute the optimal phase shift from $\bm{x}$. We show that one can compute $\bm{x}$ using the received power from three different measurements. By expanding the function $f{\left(\varphi_{l}\right)}$, we get $f{\left(\varphi_{l}\right)}\!=\!\left|z_{0}\right|^{2}\!\!+\! \left|z\right|^{2}\!\!+\!2\operatorname{Re}\left(z_{0}z^{*}\right)\!\cos{\left(\varphi_{l}\right)}\!\!+\!2\operatorname{Im}\left(z_{0}z^{*}\right)\!\sin{\left(\varphi_{l}\right)}=\bm{a}_{l}^\mathsf{T}\bm{x}$,
where $\bm{a}_{l}\delequal\left[1, \cos{\left(\varphi_{l}\right)},  \sin{\left(\varphi_{l}\right)}\right]^{\mathsf{T}}$, for $1\leq l\leq 3$. 
Assuming $\bm{A}\delequal \left[\bm{a}_{1}, \bm{a}_{2}, \bm{a}_{3}\right]^\mathsf{T}$, and $\bm{y}=\left[f{\left(\varphi_{1}\right)}, f{\left(\varphi_{2}\right)}, f{\left(\varphi_{3}\right)}\right]^{\mathsf{T}}$, we have $\bm{A}\bm{x}=\bm{y}$, or $\bm{x}=\bm{A}^{-1}\bm{y}$ for $\det\!{\left(\bm{A}\right)}\neq 0$.\qed
\section{Proof of Theorem \ref{Theorem_3}}\label{App_T3}
We denote the phase-shift vector generated by the algorithm at iteration $m$ up to the element $n$ with $\bm{\vartheta}^{k}$, where $k=mN+n$. We prove that for a given $N$, $\lim_{k\to\infty} f{\left(\bm{\vartheta}^{k}\right)}=\max_{\bm{\vartheta}}f{\left(\bm{\vartheta}\right)}$.

For any integer $1\leq n\leq N$ and an integer $k\geq 1$, we define $w_{n}^{k}\delequal \sum_{\substack{i=1\\i\neq n}}^{N}z_{i}e^{j\vartheta^{k}_{i}}$.
Hence, we have $f{\left(\bm{\vartheta}^{k}\right)} \!=\! \left|w^{k}_{n}\!+\!z_{n}e^{j\vartheta^{k}_{n}}\right|^2$.
At the phase update $k\!+\!1$, the algorithm only updates the phase of the element $n$ ($w^{k+1}_{n}\!=\!w^{k}_{n}$) as $f{\left(\bm{\vartheta}^{k+1}\right)} \!=\! \left|w^{k}_{n}\!+\!z_{n}e^{j\vartheta^{k+1}_{n}}\right|^2$.
According to the Theorem \ref{Theorem_2}, we have
\begin{equation}
    f{\left(\!\bm{\vartheta}^{k+1}\!\right)} \!=\! \left|w^{k}_{n}\!+\!z_{n}e^{j\vartheta^{k+1}_{n}}\!\right|^2\!\geq\! \left|w^{k}_{n}\!+\!z_{n}e^{j\vartheta^{k}_{n}}\!\right|^2\!=\!f{\left(\!\bm{\vartheta}^{k}\!\right)}.
\end{equation}
Therefore, $f{\left(\bm{\vartheta}^{1}\right)}, f{\left(\bm{\vartheta}^{2}\right)}, \dots$ form an increasing sequence. Moreover,  the maximum value of $f{\left(\cdot\right)}$ is $\left(\sum_{n=1}^{N}\left|z_{n}\right|\right)^{2}$, hence, the set $\mathcal{F}\delequal\left\{f{\left(\bm{\vartheta}^{k}\right)}: k\in\mathbb{N}\right\}$ is upper bounded by $\left(\sum_{n=1}^{N}\left|z_{n}\right|\right)^{2}$. Therefore, according to the monotone convergence theorem \cite{bogachev2007measure}, we have $\lim_{k\to\infty}f{\left(\bm{\vartheta}^{k}\right)} = \sup{\mathcal{F}}\leq \left(\sum_{n=1}^{N}\left|z_{n}\right|\right)^{2}$.
If we show that $\sup{\mathcal{F}}=\left(\sum_{n=1}^{N}\left|z_{n}\right|\right)^{2}$, the proof is complete. Lets define $\bm{\vartheta}^{\star}\delequal\lim_{k\to\infty}\bm{\vartheta}^{k}$. For the phase shift vector $\bm{\vartheta}^{\star}$, any further phase update will not increase $f{\left(\cdot\right)}$. In other words, for updating element $n$, we should apply Theorem \ref{Theorem_2} to the function $f{\left(\bm{\vartheta}^{\star}\right)} = \left|w_{n}^{\star}+z_{n}e^{j\vartheta^{\star}_{n}}\right|^2$,
where $w_{n}^{\star}\delequal \sum_{\substack{i=1\\i\neq n}}^{N} z_{i}e^{j\vartheta^{\star}_{i}}$. Since any further update will not increase the value of $f{\left(\cdot\right)}$, therefore for all $1\leq n\leq N$, we have $\operatorname{Arg}{\left(z_{n}e^{j\vartheta^{\star}_{n}}\right)}=\operatorname{Arg}{\left(w_{n}^{\star}\right)}$. Using Lemma \ref{lemma_8}, we have $\operatorname{Arg}{\left(z_{1}e^{j\vartheta^{\star}_{1}}\right)}=\operatorname{Arg}{\left(z_{2}e^{j\vartheta^{\star}_{2}}\right)}=\dots = \operatorname{Arg}{\left(z_{N}e^{j\vartheta^{\star}_{N}}\right)}$, or $\operatorname{Arg}{\left(z_{1}\right)}+\vartheta^{\star}_{1}=\dots= \operatorname{Arg}{\left(z_{N}\right)}+\vartheta^{\star}_{N} = \vartheta_{0} \!\mod 2\pi$.
Therefore, we have $\vartheta_{n}^{\star}=\vartheta_{0}-\arg{\left(z_{n}\right)}, \ \text{for all}\ 1\leq n\leq N$,
that are the phase shifts that maximize $f{\left(\cdot\right)}$. \qed

\begin{lem} \label{lemma_8}
    Assume for each $1\leq n\leq N$, $u_{n}\delequal \sum_{\substack{i=1\\i\neq n}}^{N}z_{i}$, if 
    $\operatorname{Arg}{\left(z_{n}\right)}=\operatorname{Arg}{\left(u_{n}\right)}$ for all $1\leq n\leq N$, then
    $\operatorname{Arg}{\left(z_{1}\right)}=\operatorname{Arg}{\left(z_{2}\right)}=\dots=\operatorname{Arg}{\left(z_{N}\right)}$.
\end{lem}
\begin{proof}
    For $1\leq m, n\leq N$ and $m\neq n$, we have $\operatorname{Arg}{(z_{m})}=\operatorname{Arg}{\left(u_{m}\right)}$ and $\operatorname{Arg}{\left(z_{n}\right)}=\operatorname{Arg}{\left(u_{n}\right)}$, therefore, for some real positive $c_{m}$ and $c_{n}$, we have $z_{m}=c_{m}u_{m}$ and $z_{n}=c_{n}u_{n}$ .
    Assuming $u_{m,n}\delequal \sum_{\substack{i=1\\i\neq m,n}}^{N}z_{k}$, we have $z_{m}=c_{m}\left(z_{n}+u_{m,n}\right)$ and $z_{n}=c_{n}\left(z_{m}+u_{m,n}\right)$.
    After some algebraic manipulations, we obtain $z_{m}\!=\!\tfrac{c_{m}\left(1+c_{n}\right)}{c_{n}\left(1+c_{m}\right)}z_{n}$. 
    Hence, $\operatorname{Arg}{\left(z_{m}\right)}\!=\!\operatorname{Arg}{\left(z_{n}\right)}$ and the proof is complete.
\end{proof}
\section{Proof of Theorem \ref{Theorem_4}} \label{App_T4}
We  know $\vartheta^{\star}=\arg\left(z_{0}\right)-\arg\left(z\right)$ is the optimal variable for the continuous domain. For a discrete domain $\Omega=\left\{\omega_{1}, \omega_{2}, \dots, \omega_{\left|\Omega\right|}\right\}$, the variable $\vartheta=\omega_{k^{\star}}$ maximizes $f{\left(\cdot\right)}$, where
\begin{align}
    k^{\star}\!&=\underset{1\leq k\leq \left|\Omega\right|}{\mathrm{argmax}}\  {\left|z_{0}+ze^{j\omega_{k}}\right|^{2}}=\underset{1\leq k\leq \left|\Omega\right|}{\mathrm{argmax}}\  {\operatorname{Re}\!{\left(z_{0}z^{*}e^{-j\omega_{k}}\right)}}\nonumber\\
    &=\underset{1\leq k\leq \left|\Omega\right|}{\mathrm{argmax}}\  {\operatorname{Re}\!{\left(\left|z_{0}z^{*}\right|e^{j\left(\vartheta^{\star}-\omega_{k}\right)}\right)}}\!=\underset{1\leq k\leq \left|\Omega\right|}{\mathrm{argmax}}\ \cos{\left(\vartheta^{\star}\!-\omega_{k}\right)}\nonumber\\ 
    &=\underset{1\leq k\leq \left|\Omega\right|}{\mathrm{argmax}}\  \cos{\left(\zeta_{k}\right)}=\underset{1\leq k\leq \left|\Omega\right|}{\mathrm{argmax}}\  \cos{\left(\min\left(\zeta_{k}, 2\pi-\zeta_{k}\right)\right)}\nonumber\\ 
    &=\underset{1\leq k\leq \left|\Omega\right|}{\mathrm{argmin}}\  \min\left(\zeta_{k}, 2\pi-\zeta_{k}\right), 
\end{align}
where $\zeta_{k}\delequal \vartheta^{\star}-\omega_{k}\mod 2\pi$.\qed
\section{Proof of Theorem \ref{Theorem_5}} \label{App_T5}

Algorithm \ref{Alg4} intends to maximize $f{\left(\bm{\vartheta}\right)}=|\sum_{i=1}^{N}z_{i}e^{j\vartheta_{i}}|^{2}$ with discrete phase-shifts. We denote the phase-shift vector generated by the algorithm at iteration $m$ up to the element $n$ with $\bm{\vartheta}^{k}$, where $k=mN+n$. We will prove that the sequence $f{\left(\bm{\vartheta}^{1}\right)}, f{\left(\bm{\vartheta}^{2}\right)}, \dots$ converges.

For any integer $1\leq n\leq N$ and an integer $k\geq 1$, we define $w_{n}^{k}\delequal \sum_{\substack{i=1\\i\neq n}}^{N}z_{i}e^{j\vartheta^{k}_{i}}$.
Hence, we have $f{\left(\bm{\vartheta}^{k}\right)} \!=\! |w^{k}_{n}\!+\!z_{n}e^{j\vartheta^{k}_{n}}|^2$. At the phase update $k\!+\!1$, the algorithm only updates the phase of the element $n$ ($w^{k+1}_{n}\!=\!w^{k}_{n}$) as $f{\left(\bm{\vartheta}^{k+1}\right)} \!=\! |w^{k}_{n}\!+\!z_{n}e^{j\vartheta^{k+1}_{n}}|^2$. According to the Theorem \ref{Theorem_4}, we have 
\begin{equation}
    f{\left(\bm{\vartheta}^{k+1}\right)} \!=\! \left|w^{k}_{n}\!+\!z_{n}e^{j\vartheta^{k+1}_{n}}\!\right|^2\!\geq\! \left|w^{k}_{n}\!+\!z_{n}e^{j\vartheta^{k}_{n}}\!\right|^2\!=\!f{\left(\bm{\vartheta}^{k}\right)}.
\end{equation}
Therefore, $f{\left(\bm{\vartheta}^{1}\right)}, f{\left(\bm{\vartheta}^{2}\right)}, \dots$ form a monotonically increasing sequence. Moreover, the maximum value of $f{\left(\cdot\right)}$ is $(\sum_{n=1}^{N}\left|z_{n}\right|)^{2}$, hence, the set $\mathcal{F}\delequal\left\{f{\left(\bm{\vartheta}^{k}\right)}: k\in\mathbb{N}\right\}$ is upper bounded by $(\sum_{n=1}^{N}|z_{n}|)^{2}$. Therefore, according to the monotone convergence theorem \cite{bogachev2007measure}, we have $\lim_{k\to\infty}f{\left(\bm{\vartheta}^{k}\right)} = \sup{\mathcal{F}}\leq (\sum_{n=1}^{N}\left|z_{n}\right|)^{2}$. \qed

\bibliographystyle{IEEEtran}
\bibliography{IEEEabrv,ref}
\end{document}